\documentclass[journal]{IEEEtran}
\usepackage{amssymb,amsmath,amsthm,enumitem}

\usepackage{graphicx}
\usepackage[font=footnotesize]{caption}
\usepackage{subcaption}
\usepackage{algorithm}
\usepackage{algorithmicx}
\usepackage[labelformat=parens,labelsep=quad,skip=3pt]{caption}

\ifCLASSINFOpdf
\else
\fi
\usepackage{amsmath}
\usepackage{amssymb}  
\usepackage{amsfonts}
\usepackage{mathrsfs}
\usepackage{mathptmx} 
\usepackage{mathtools}

\usepackage{graphicx}
\usepackage{textcomp}
\usepackage{xcolor}
\usepackage{subcaption}
\usepackage{placeins}

\usepackage{algorithm}
\usepackage{algorithmicx}
\usepackage{algpseudocode}

\usepackage{multirow}
\usepackage{subcaption}
\usepackage{csquotes}

\usepackage[hidelinks]{hyperref}

\usepackage{cite}

\newtheorem{assumption}{Assumption}
\newtheorem{lemma}{Lemma}
\newtheorem{remark}{Remark}
\newtheorem{theorem}{Theorem}
\newtheorem{corollary}{Corollary}
\newtheorem{definition}{Definition}
\newtheorem{problem}{Problem}

\begin{document}
%
\title{Data-Driven Motion Planning for Uncertain Nonlinear Systems}

%


\author{%
  Babak~Esmaeili,%
  ~Hamidreza~Modares\textsuperscript{*},%
  ~and~Stefano~Di~Cairano%
  \thanks{This work is supported by the National Science Foundation under award ECCS‑2227311.}%
  \thanks{B.~Esmaeili and H.~Modares are with the Department of Mechanical Engineering, Michigan State University, East Lansing, MI 48823, USA. (e‑mails: \texttt{esmaeil1@msu.edu}; \texttt{modaresh@msu.edu}).}%
  \thanks{S.~Di~Cairano is with Mitsubishi Electric Research Laboratories (MERL), Cambridge, MA 02139, USA. (e‑mail: \texttt{dicairano@merl.com}).}%
}


%
%
\maketitle

\begin{abstract}
This paper proposes a data‑driven motion‑planning framework for nonlinear systems that constructs a sequence of overlapping invariant polytopes. Around each randomly sampled waypoint, the algorithm identifies a convex admissible region and solves data‑driven linear‑matrix‑inequality problems to learn several ellipsoidal invariant sets together with their local state‑feedback gains. The convex hull of these ellipsoids—still invariant under a piece‑wise‑affine controller obtained by interpolating the gains—is then approximated by a polytope. Safe transitions between nodes are ensured by verifying the intersection of consecutive convex‑hull polytopes and introducing an intermediate node for a smooth transition. Control gains are interpolated in real time via simplex‑based interpolation, keeping the state inside the invariant polytopes throughout the motion. Unlike traditional approaches that rely on system dynamics models, our method requires only data to compute safe regions and design state-feedback controllers. The approach is validated through simulations, demonstrating the effectiveness of the proposed method in achieving safe, dynamically feasible paths for complex nonlinear systems.
\end{abstract}


\IEEEpeerreviewmaketitle

\section{Introduction}

\IEEEPARstart{M}{}otion planning is a fundamental problem in control and robotics, which involves determining a feasible trajectory for a system to move from an initial state to a desired target state while avoiding obstacles and satisfying system constraints \cite{gonzalez2015review}.  Over the years, several motion-planning approaches have been proposed, including graph search-based methods \cite{likhachev2003ara}, sampling-based methods like rapidly exploring random trees (RRT) \cite{lavalle2001randomized}, behavior-based approaches \cite{mataric1992behavior}, machine learning-based approaches \cite{sun2021motion}, potential fields \cite{ge2002dynamic}, and optimization-based techniques such as differential dynamic programming \cite{li2020hybrid}. Among them, RRT, as a sampling-based approach, has received a surge of interest due to its success in robotic applications. However, most of these successful strategies are under assumptions that cannot be certified in many applications \cite{lavalle2006planning, karaman2011sampling}. For instance, the planning is typically performed assuring that the waypoints are kinematically feasible. In ever-changing environments, dynamic feasible of the trajectories is also crucial. Another challenge is that the system dynamics are nonlinear. If RRT planner does not account for nonlinear dynamics, frequent re-planning might be required, ruining the system performance.  Motion planning for nonlinear autonomous systems is a challenging problem due to the complexity of their dynamics and the need to account for multiple factors, such as system constraints and safety requirements. Unlike linear systems, the inherent nonlinearities in such systems make traditional planning methods less effective or computationally expensive \cite{frazzoli2005maneuver}. 

Another challenge is that system dynamics can be uncertain, which requires leveraging data-driven approaches to reduce uncertainties. Sampling-based planners generate random waypoints without verifying if the system is capable of safely traversing them, given only available data. This can lead to unsafe maneuvers and potential collisions when executed on real systems \cite{karaman2011sampling, kavraki1996probabilistic}. {Therefore, integrating data-driven safety guarantees into motion-planning algorithms is essential to ensure reliable operation \cite{schouwenaars2001mixed}. One promising direction is to leverage motion planners based on finding a sequence of overlapping invariant sets \cite{weiss2017motion,danielson2020robust,danielson2016path,niknejad2024soda} and extend them to data-based planners rather than model-based and to nonlinear systems rather than linear systems. 
Constraint admissible sets are regions in the state space within which the system can remain safe under certain control laws, even in the presence of disturbances and model uncertainties. By learning a sequence of invariant sets from data, planners can generate only those waypoints and trajectories that are guaranteed to remain within these safe regions, given available data, thus reducing the risk of unsafe behavior during execution. In this framework, an edge is created only if the current node lies inside the ellipsoid of the next node, ensuring that switching between controllers preserves safety. While these criteria ensure safe transitions, they are inherently conservative and can significantly limit connectivity—especially in narrow, cluttered, or high-dimensional environments. This often leads to sparse graphs and missed feasible transitions. Moreover, these methods are largely restricted to linear systems with known dynamics, where computing ellipsoidal invariant sets via LMIs is tractable \cite{greiff2024invariant}. Extending such approaches to nonlinear systems is substantially more challenging, as invariant sets are harder to compute. These combined limitations highlight the need for more flexible and scalable frameworks—particularly those that integrate data-driven modeling with control-theoretic safety guarantees—to enable reliable motion planning in complex real-world environments.

A common strategy for ensuring safety is through set-theoretic control design, which leverages $\lambda$-contractivity to guarantee that a given set remains invariant for the closed-loop system. This method ensures that, starting from within the set, the system’s states do not leave it, thus maintaining safety \cite{blanchini1999set}. However, as the complexity of the system or the admissible set grows, it becomes increasingly challenging to make the entire admissible set invariant. In practice, admissible sets represent regions where the system states are allowed to evolve, often defined by physical limitations and environmental constraints. Designing controllers that make complex admissible sets invariant is a difficult task, and the resulting invariant set is typically a subset of the admissible set, whose size depends on the richness of the available data and the control structure \cite{bisoffi2020data,de2021low}.  Partitioning complex admissible sets into smaller, disjoint subsets is a promising solution for overcoming this limitation. These partitioning-based methods enable the design of controllers that ensure safety within each partition, especially when linear feedback alone cannot guarantee invariance \cite{nguyen2023,nguyen2024}. However, these methods are typically limited to systems with known dynamics since they require an accurate model of the system. In practical scenarios where the system model is unavailable or difficult to obtain, there is a growing need for data-driven approaches that can ensure safety using only available data, without relying on a predefined model.

In recent years, data-driven control strategies have been explored extensively for achieving autonomy in complex systems. These strategies can be broadly categorized into indirect and direct methods \cite{hou2013model}. Indirect methods first identify system models from data and then use these models for controller design \cite{wang2016indirect}. In contrast, direct methods learn control policies directly from data without the need for system modeling \cite{bisoffi2022controller}. Direct data-driven approaches have shown great success in linear time-invariant (LTI) systems, avoiding the inaccuracies introduced during model identification \cite{modares2023data}. However, extending these methods to nonlinear systems is far more challenging due to the complex dynamics and nonlinearities involved. Recent works have proposed semidefinite programming (SDP)-based techniques to achieve stabilizing, direct control for nonlinear systems using finite data \cite{luppi2022data}. These methods, while promising, often require canceling all nonlinearities, which can increase control effort or ignore beneficial nonlinear effects \cite{guo2023learning}. Moreover, they typically do not incorporate safety guarantees, limiting their effectiveness in practical applications where ensuring safety is crucial \cite{de2023learning,guo2024data}.

In this paper, we propose a novel data-driven motion-planning algorithm that ensures safety for nonlinear systems by leveraging overlapping invariant convex hulls of ellipsoidal sets. Unlike traditional methods, our approach does not require a known system model or explicit state-space representation. Instead, it relies on a purely data-driven framework to compute control gains and generate safe transitions between sampled states within the admissible region. The key contributions of this paper are summarized as follows:
\begin{itemize}
    \item We introduce a data-driven motion-planning algorithm that guarantees safety by constructing convex hulls of ellipsoidal sets for nonlinear systems.
    \item We develop a systematic method to check for safe transitions between sampled states by verifying the intersection of invariant sets.
    \item We propose an efficient feedback control strategy that ensures the system remains within the admissible region while driving it toward the target state and minimizing the nonlinear residuals.
\end{itemize}

The proposed approach is well-suited for nonlinear systems with unknown dynamics and offers a framework for safe motion planning in complex environments. We demonstrate the effectiveness of our algorithm through extensive simulations, highlighting its ability to maintain safety in challenging scenarios.

\noindent \textbf{Notations}: Throughout this paper, the identity matrix of size \( n \times n \) is denoted by \( I_n \), and \( \mathbf{0}_n \) represents the \( n \times n \) zero matrix. The set of real symmetric \( d \times d \) matrices is denoted by \( \mathbb{S}^d \). For any matrix \( A \), \( A_i \) refers to its \( i \)-th row, and \( A_{ij} \) denotes the element in the \( i \)-th row and \( j \)-th column. When matrices or vectors \( A \) and \( B \) have the same dimensions, the notation \( A (\leq, \geq) B \) represents elementwise inequality, i.e., \( A_{ij} (\leq, \geq) B_{ij} \) for all \( i \) and \( j \).

The trace, largest eigenvalue, and transpose of a matrix \( A \) are denoted by \( \operatorname{Tr}(A) \), \( \lambda_{\max}(A) \), and \( A^\top \), respectively. The spectral and Frobenius norms of \( A \) are denoted by \( |A|_2 \) and \( |A|_F \). For symmetric matrices, the symbol \( (*) \) is used to indicate the symmetric completion of a block to preserve matrix symmetry. For a matrix \( Q \), the notation \( Q (\preceq, \succeq) 0 \) indicates that \( Q \) is negative or positive semi-definite, respectively. For a set \( \mathcal{S} \) and a non-negative scalar \( \mu \), the scaled set \( \mu \mathcal{S} \) consists of all elements \( \mu x \), where \( x \in \mathcal{S} \).

A directed graph is denoted by \( G = (V, E) \), where \( V \) is a finite set of vertices and \( E \) is a set of directed edges. Each edge \( (u, v) \in E \) is an ordered pair, where \( u \) is the tail and \( v \) is the head, indicating the direction of traversal. A path in \( G \) is an ordered sequence of vertices such that each consecutive pair is connected by a directed edge, preserving directionality. The goal of a graph search algorithm is to identify a valid path between specified vertices while satisfying predefined constraints.

The convex hull generated by sets \( \mathcal{S}_1, \mathcal{S}_2, \ldots, \mathcal{S}_n \) is denoted as \( \mathcal{S} = \operatorname{Co}(\mathcal{S}_1, \mathcal{S}_2, \ldots, \mathcal{S}_n) \). Any point \( x \in \mathcal{S} \) can be written as a convex combination of elements \( x_i \in \mathcal{S}_i \), specifically:
\begin{equation}\label{eq.combination}
x = \alpha_1 x_1 + \alpha_2 x_2 + \cdots + \alpha_n x_n,
\end{equation}
where \( \alpha_i \in [0, 1] \) for all \( i = 1, \dots, n \), and \( \sum_{i=1}^n \alpha_i = 1 \).

\begin{definition}\label{def_0}
For any two positive integers \( a \) and \( b \), the operation \( \operatorname{mod}(a,b) \) returns the remainder when \( a \) is divided by \( b \). Consider a finite set with \( M \) elements. The rotational indexing function \( \operatorname{R_M}(i) \) maps index \( i \) to a new index \( j \) in a cyclic manner. In this work, the mapping is defined as:
\begin{equation}\label{eq.mapping}
j = \operatorname{R_M}(i) = \operatorname{mod}(i + M - 2, M) + 1.
\end{equation}
\end{definition}


\begin{lemma}[\cite{poursafar2010model}]\label{lemma_1} 
Let $M$ and $N$ be real constant matrices, and let $P$ be a symmetric positive definite matrix. For any scalar $\epsilon > 0$ and $\mu \geq \lambda_{\max}(P)$, the following inequality holds \begin{equation} 
\begin{aligned} M^\top P N + N^\top P M & \leq \epsilon M^\top P M + \epsilon^{-1} N^\top P N \\ 
& \leq \epsilon M^\top P M + \epsilon^{-1} \mu N^\top N. 
\end{aligned} 
\end{equation} 
\end{lemma}

The following definitions are used to describe admissible and safe sets in this paper.


\begin{definition}\label{def_2}
A polytope is the intersection of a finite number of half-spaces and is expressed as 
\begin{equation}\label{eq.polyhedral_set}
    \mathcal{S} = \{ x \in \mathbb{R}^n \mid F x \leq g \},
\end{equation}
where $F \in \mathbb{R}^{m \times n}$ and $g \in \mathbb{R}^m$ define the constraints of the polytope.
\end{definition}

\begin{definition}\label{def_3}
An ellipsoidal set, denoted by $\mathcal{E}(P, c)$, is defined as  
\begin{equation}
    \mathcal{E}(P, c) := \{ x \in \mathbb{R}^n \mid (x - c)^\top P^{-1} (x - c) \leq 1 \},
\end{equation}
where $P^{-1} \in \mathbb{R}^{n \times n}$ is a symmetric positive definite matrix, and $c \in \mathbb{R}^n$ represents the center of the ellipsoid.
\end{definition}

\section{Problem Statement}
Consider the following nonlinear system
\begin{align}\label{eq.system}
x_{k+1} &= f(x_k, u_k) \nonumber \\
y_k &= Cx_k,
\end{align}
where \(x_k\in\mathbb{R}^{n}\) and \(u_k\in\mathbb{R}^{m}\) are the state and input at time \(k\), respectively, and \(f(\cdot,\cdot)\) represents the unknown nonlinear dynamics. The output \(y_k\in\mathbb{R}^{2}\) is defined as the subset of state components that correspond to the system's 2-D position, i.e.,\ \(y_k = x_{\text{pos},k}\), and 
\(C=\bigl[I_{2}\;\;0_{2\times(n-2)}\bigr]\).

\begin{assumption}\label{assumption_2}
The nonlinear function $f(x_k,u_k)$ satisfies the Lipschitz condition $f(x_k,u_k)^\top f(x_k,u_k) \leq [x^\top_k u^\top_k] Q^\top Q [x^\top_k u^\top_k]^\top$, where $Q$ is a known constant matrix with compatible dimensions, though it is not necessarily nonsingular.
\end{assumption}


\begin{definition}[Admissible Set]\label{definition_6}
Consider the discrete‑time nonlinear system \eqref{eq.system} subject to state‑constraint functions \(g(x)\le 0\).  
The \emph{admissible set} is
\begin{equation}
\mathcal{X}\triangleq\bigl\{\,x\in\mathbb{R}^{n}\;\bigl|\;g(x)\le 0\bigr\},
\end{equation}
i.e., the collection of all states that satisfy every physical and operational constraint of the system.
\end{definition}


\begin{definition}[Safe Set]\label{definition_7}
Let the closed‑loop system be
\begin{equation}
x_{k+1}=f_\text{cl}(x_k)\quad\text{with}\quad
    f_\text{cl}(x)\triangleq f\bigl(x,K(x)\bigr),
\end{equation}
and let \(\mathcal{X}\) be the admissible set from Definition \ref{definition_6}.  
A set \(\mathcal{S}\subseteq\mathcal{X}\) is called a safe set if
\begin{equation}
x\in\mathcal{S}\;\Longrightarrow\;f_\text{cl}(x)\in\mathcal{S}.
\end{equation}

That is, once the state enters \(\mathcal{S}\), it remains there for all future time steps under the implemented feedback controller \(u=K(x)\).
\end{definition}

\begin{assumption}\label{Assumption_5} 
The admissible set is assumed to be a time-invariant polyhedral set denoted by ${\cal{S}}(F, g)$, as described in \eqref{eq.polyhedral_set}. 
\end{assumption}

\begin{assumption}
The state \( x_k \) is subject to constraints
\begin{equation}
    x_k \in \mathcal{X}.
\end{equation}
The state constraint set \( \mathcal{X} \subseteq \mathbb{R}^n \) is generally non-convex; however, it can be represented as the union of convex subsets
\begin{equation}
    \mathcal{X} = \bigcup_{\kappa \in \mathcal{I}_\mathcal{X}} \mathcal{X}_\kappa,
\end{equation}
where \( \mathcal{I}_\mathcal{X} \) is a finite index set (\(|\mathcal{I}_\mathcal{X}| < \infty\)), and each subset \( \mathcal{X}_\kappa \subseteq \mathbb{R}^n \) is a compact polytope defined by a set of linear inequalities
\begin{equation}
    \mathcal{X}_\kappa(F,g) = \left\{ x \mid F_{\mathcal{X}_\kappa}(l)^T x \leq g_{\mathcal{X}_\kappa}(l), \quad l = 1, \dots, |\mathcal{I}_\mathcal{X}| \right\}.
\end{equation}
\end{assumption}



We isolate the planar position by letting
\begin{equation}
x_{\text{pos},k}
\;=\;
\begin{bmatrix}x_{1,k}\\[2pt]x_{2,k}\end{bmatrix}\in\mathbb{R}^{2},
\end{equation}
so, the output \(y_k\) is exactly the position vector already denoted by
\(x_{\text{pos},k}\).

\begin{definition}[Projected admissible region]\label{def:Xpos}
Given the full‑state constraint set \(\mathcal{X}\subset\mathbb{R}^{n}\),
the \emph{projected admissible region} in the 2‑D position subspace is
\[
    \mathcal{X}_{\text{pos}}
    \;=\;
    \bigl\{\,Cx\in\mathbb{R}^{2}\mid x\in\mathcal{X}\bigr\},
\]
i.e., the image of \(\mathcal{X}\) under the projection matrix
\(C\).
\end{definition}

This paper presents a data-driven motion-planning algorithm for nonlinear systems that ensures safe navigation by steering the position state \( x_{pos,k} \) from an initial point to a target through overlapping invariant convex hulls of ellipsoids. These ellipsoidal sets are data-driven, lie within the non-convex admissible region \( \mathcal{X} \), and guarantee safety and constraint satisfaction during transitions.

Unlike traditional planners based on Euclidean distance, this method leverages state measurements and set-based conditions to determine safe transitions. The algorithm computes a sequence of waypoints and corresponding control gains that respect system dynamics while ensuring safety.

We formally describe the problem as follows.

\begin{problem}\label{prob:safe_rrt}
Consider the discrete-time nonlinear system described by \eqref{eq.system}. The goal is to find a sequence of waypoints $\{p_0,p_1,\dots,p_{N_w}\}$ and nonlinear state-feedback control gains such that
\[
  p_{N_w}\in\mathcal{B}_{r_f}(x_{pos,f})
  \;=\;\bigl\{x\in\mathbb{R}^n \,\bigl|\,\|x_{pos}-x_{pos,f}\|\le r_f\bigr\},
\]
where \(r_f>0\) is a user‑specified termination radius that defines the acceptable neighbourhood around the target position \(x_{pos,f}\). At all times the system state must remain within the admissible region $\mathcal{X}$. The steps to achieve this are described as follows.
\end{problem}

\begin{enumerate}
    \item \textbf{Sample and certify a candidate waypoint \(w_i\)}: Randomly sample the admissible region \(\mathcal{X}_{\text{pos}}\) to generate a candidate waypoint \(p_i\).  
      Using the pre‑collected state–input data set \(\mathcal{D}=\{(x_j,u_j)\}_{j=1}^N\) gathered offline during representative operating runs, solve the data‑driven LMIs to obtain a local feedback gain and construct the invariant set \(\mathcal{C}_i\) (a convex hull of ellipsoids) that guarantees safety and constraint satisfaction.
    
    \item \textbf{Check for safe transitions}: Verify whether the current invariant set overlaps with at least one previously accepted set. If an overlap exists, compute the intermediate safe point and add the candidate waypoint to the graph. Otherwise, discard the candidate.
    
    \item \textbf{Repeat until convergence}: Iteratively sample and certify candidate waypoints, and check for safe transitions from the current node to the new waypoint. Continue this process to build a sequence of connected invariant sets until a waypoint satisfying \(p_{N_w}\in\mathcal{B}_{r_f}(x_{pos,f})\) is reached.
\end{enumerate}




By iteratively constructing invariant sets, the algorithm ensures safety and feasibility at each planning step. Transitions between waypoints are verified through convex hull intersections, enabling smooth and reliable progression toward the goal.

Unlike conventional motion-planning methods that rely on predefined models, this approach is entirely data-driven, leveraging measured state information to handle nonlinearities and non-convex constraints during both planning and control design. This makes it particularly suitable for systems with complex or unknown dynamics.

Theorem~\ref{thm:path_existence} provides sufficient conditions for the existence of a feasible path that solves Problem~1. It links the solvability of the motion planning problem to the existence of a valid path in the constructed graph representation based on invariant sets.

\begin{theorem}[\cite{danielson2016path}]\label{thm:path_existence}
Consider a motion planning problem where the system state evolves within invariant sets. Construct a graph \( G = (V, E) \) by connecting waypoints obtained through successive solutions of Problem 1. If there exists a path in \( G \) that originates from the initial waypoint \( p_0 = x_{pos,0} \) and reaches a node \(p_{N_w}\in\mathcal{B}_{r_f}(x_{\mathrm{pos},f})\), then Problem 1 is solvable for the prescribed neighbourhood radius \(r_f\).
\end{theorem}


In the following sections, we present a complete data-driven framework to solve the safe motion planning problem for nonlinear systems introduced in Problem~1. Section III introduces the data-driven representation of the nonlinear dynamics using lifted coordinates and integral control augmentation. Section IV formulates the safe control problem using invariant ellipsoidal sets and provides the necessary LMIs to compute safe controllers. Section V extends this method to construct convex hulls of ellipsoids for less conservative and more flexible motion planning. Each section progressively builds toward enabling safe transitions between waypoints in a data-driven setting, culminating in a motion planning algorithm that guarantees safety without requiring an explicit model of the system dynamics.

\section{Data-Driven Representation}

Inspired by the work in \cite{de2023learning}, we aim to represent the system in a purely data-driven framework without relying on explicit system models. By leveraging available data, we directly represent the system dynamics and design suitable control strategies. While \cite{de2023learning} addresses set-point tracking for continuous-time control-affine systems, we extend this framework to the motion-planning problem for general discrete-time nonlinear (non-affine) systems, where the system must safely track a sequence of waypoints generated during planning.

We define the stacked vector of states and inputs as $\xi := \begin{bmatrix} x \\ u \end{bmatrix}$ and consider the following assumption.

\begin{assumption}
A continuously differentiable function $Z: \mathbb{R}^{n+m} \to \mathbb{R}^{n_z}$ is given, such that $f(x, u) = A Z(\xi)$ for some matrix $A \in \mathbb{R}^{n \times n_z}$.
\end{assumption}

Using Assumption 3, \eqref{eq.system} is written as follows
\begin{equation}
x_{k+1} = AZ(\xi_k)
\end{equation}
where
\begin{equation} 
Z(\xi_k) = \begin{bmatrix} \xi_k \\ S(\xi_k) \end{bmatrix}, 
\end{equation}
with $S(\xi_k)$ serving as a dictionary of functions that approximates the system dynamics. Also, $A = \begin{bmatrix} \bar{A} & \hat{A} \end{bmatrix}$
with \, $\bar{A} \in \mathbb{R}^{n \times (n+m)}$ \, and \, $\hat{A} \in \mathbb{R}^{n \times (n_z-n-m)}$.

We redefine $\xi$ as the new state variable and incorporate integral control. Specifically, we introduce a new control input $v_k \in \mathbb{R}^m$ and modify the original input dynamics through an input‑increment formulation. Hence, the input update is given by
\begin{equation} 
u_{k+1} = u_k + T_s v_k, 
\end{equation} 
where $T_s$ denotes the sampling time. Consequently, the augmented dynamics, incorporating the input‑increment, become
\begin{equation} 
\xi_{k+1} = \begin{bmatrix} \bar{A} & \hat{A} \\
\begin{bmatrix} 0 & I \end{bmatrix} & 0 \end{bmatrix} \begin{bmatrix} \xi_k \\ S(\xi_k) \end{bmatrix} + \begin{bmatrix} 0 \\ T_s I \end{bmatrix} v_k. \end{equation}

Since the goal is to design a motion planner for tracking a desired waypoint $r_k \in \mathbb{R}^2$, we introduce the tracking error as
\begin{equation} 
e_k := y_k - r = x_{pos,k} - r_k. 
\end{equation}

The objective is to minimize this error while ensuring that the system remains within a safe set throughout its evolution. To achieve accurate tracking, we incorporate another integral action, which helps eliminate steady-state errors. We define the augmented state $\zeta$ that includes the integral of the tracking error as well
\begin{equation} 
\zeta_k := \begin{bmatrix} \xi_k \\ \eta_k 
\end{bmatrix}, \quad \eta_{k+1} = \eta_k + T_se_k,
\end{equation}
where $\eta_k \in \mathbb{R}^2$ represents the integral of the error.

Define
\begin{align}
\mathcal{Z}(\xi_k, \eta_k) := \begin{bmatrix} \xi_k \\ \eta_k \\ S(\xi_k) \end{bmatrix} \in \mathbb{R}^{n_{\mathcal{z}}}.
\end{align}

Then, the augmented state update equations become
\begin{align}\label{eq.augmented_sys}
& \zeta_{k+1} = \begin{bmatrix}
   \xi_{k+1} \\ \eta_{k+1}
\end{bmatrix} =
\begin{bmatrix}
    \begin{bmatrix} x_{k+1} \\ u_{k+1} \end{bmatrix} \\ \eta_{k+1}
\end{bmatrix} = \nonumber \\
& \begin{bmatrix}
    \bar{A} & 0 & \hat{A} \\
    [0 \; I] & 0 & 0 \\
    [T_sC \; 0] & I & 0
\end{bmatrix}
\begin{bmatrix}
    \begin{bmatrix} x_k \\ u_k \end{bmatrix} \\ \eta_k \\ S(\xi_k)
\end{bmatrix} +
\begin{bmatrix}
    0 \\ T_sI \\ 0
\end{bmatrix}
v_k -
 \begin{bmatrix}
    0 \\ 0 \\ T_sI
\end{bmatrix}r \nonumber \\
= \; & \mathcal{A}\mathcal{Z}(\xi_k,\eta_k)+\mathcal{B}v_k-\mathcal{I}r
\end{align}


According to \eqref{eq.augmented_sys}, one has
\begin{equation}\label{eq.system_for_data}
\begin{bmatrix}
    \xi_{k+1} \\ \eta_{k}+T_sx_{pos,k}
\end{bmatrix} =\mathcal{A}\mathcal{Z}(\xi_k,\eta_k)+\mathcal{B}v_k
\end{equation}

By applying
\begin{equation}\label{eq.data_V0}
V_0 := \begin{bmatrix} v_0 & v_1 & \cdots & v_{N-1} \end{bmatrix} \in \mathbb{R}^{m \times N}
\end{equation}
to the system \eqref{eq.system_for_data}, $N+1$ samples of the state vectors are collected as follows
\begin{align}\label{eq.data_THETA}
\Xi &:= \begin{bmatrix} \xi_0 & \xi_1 & \cdots & \xi_N \\ 
\eta_0 & \eta_1 & \cdots & \eta_N \end{bmatrix} \in \mathbb{R}^{(n+m+2) \times (N+1)}
\end{align}
where these collected samples are then organized as follows 
\begin{align}\label{eq.data_Z0}
\mathcal{Z}_0 &:= \begin{bmatrix} \xi_0 & \xi_1 & \cdots & \xi_{N-1} \\ 
\eta_0 & \eta_1 & \cdots & \eta_{N-1} \\
S(\xi_0) & S(\xi_1) & \cdots & S(\xi_{N-1}) \end{bmatrix} \in \mathbb{R}^{{n_{\mathcal{z}}} \times N} \\
\Xi_1 &:= \begin{bmatrix} \xi_1 & \xi_2 & \cdots & \xi_N \\ 
\eta_1+T_sx_{pos,1} & \eta_2+T_sx_{pos,2} & \cdots & \eta_N+T_sx_{pos,N} \end{bmatrix} \\
& \in \mathbb{R}^{(n+m+2) \times N} \nonumber.
\end{align}




Given that the data sequences $V_0$, $\Xi_1$, and $\mathcal{Z}_0$ fulfill \eqref{eq.system_for_data}, one has
\begin{equation}\label{eq.closed_loop}
\Xi_1 = \mathcal{A} \mathcal{Z}_0 + \mathcal{B} V_0. 
\end{equation}

\begin{assumption}\label{assumption_3}
The data matrix $\mathcal{Z}_0$ is assumed to be full row rank, and the number of collected samples must satisfy $N \geq {n_{\mathcal{z}}} + 1$.
\end{assumption}

Consider the existence of matrices $\mathcal{K} \in \mathbb{R}^{m \times {n_{\mathcal{z}}}}$ and $\mathcal{G}_K \in \mathbb{R}^{N \times {n_{\mathcal{z}}}}$ under Assumption 4 such that
\begin{equation}\label{eq.lemma_3}
\begin{bmatrix}
\mathcal{K} \\ I_{n_z} \end{bmatrix}=\begin{bmatrix} V_0 \\ \mathcal{Z}_0\end{bmatrix} \mathcal{G}_K,
\end{equation}
where $\mathcal{G}_K = [\mathcal{G}_{K,l} \quad \mathcal{G}_{K,nl}]$ with $\mathcal{G}_{K,l} \in \mathbb{R}^{N \times n+m+2}$ and $\mathcal{G}_{K,nl} \in \mathbb{R}^{N \times ({n_{\mathcal{z}}} - n - m - 2)}$.

By multiplying both sides of \eqref{eq.closed_loop} by $\mathcal{G}_K$ from the right and using the result from \eqref{eq.lemma_3}, one gets
\begin{equation}\label{eq.closed_loop_data}
\mathcal{A} + \mathcal{B} \mathcal{K} = \Xi_1 \mathcal{G}_K. 
\end{equation}




The steady‑state dynamics of the augmented system, for a constant reference \(r\), are
\begin{equation}
\begin{bmatrix}\xi^{*}(r)\\[2pt]\eta^{*}(r)\end{bmatrix}
    =\bigl(\mathcal{A}+\mathcal{B}\mathcal{K}\bigr)\,
      \mathcal{Z}\bigl(\xi^{*}(r),\eta^{*}(r)\bigr)-\mathcal{I}\,r,
\end{equation}
where \(\xi^{*}(r)\) and \(\eta^{*}(r)\) denote the steady‑state values of the state and integral state as functions of the reference. For brevity, we omit the explicit argument \(r\) in subsequent expressions.

To analyze the tracking error dynamics, we define the error terms as
\begin{equation} 
\xi_{e,k} = \xi_k - \xi^*, \quad \eta_{e,k} = \eta_k - \eta^*. 
\end{equation}

Substituting these into the closed-loop system, we obtain
\begin{align}\label{eq.final_error_system}
\zeta_{e,k+1} &= \begin{bmatrix} \xi_{e,k+1} \\ \eta_{e,k+1} \end{bmatrix} = \begin{bmatrix} \xi_{k+1} \\ 
\eta_{k+1} \end{bmatrix} - \begin{bmatrix} \xi^* \\ \eta^* \end{bmatrix} \notag \\
&= (\mathcal{A} + \mathcal{B}\mathcal{K}) (\mathcal{Z}(\xi_k, \eta_k) - \mathcal{Z}(\xi^*, \eta^*)). 
\end{align}

Owing to \eqref{eq.closed_loop_data}, the closed-loop error dynamics are represented as the following data-based form
\begin{equation}\label{eq.data_driven_error_system} \zeta_{e,k+1} = \Xi_1 \mathcal{G}_{K,l} \zeta_{e,k} + \Xi_1 \mathcal{G}_{K,nl} (S(\xi_k) - S(\xi^*)), 
\end{equation}

This formulation captures the evolution of the tracking error using a fully data-driven framework. 

The next section formulates and solves the safe motion planning problem using invariant ellipsoids for each waypoint. This approach ensures that the system remains within predefined safety constraints while following the desired waypoint. Subsequently, the method is extended to the convex hull of ellipsoids, allowing for a more flexible and less conservative representation of the feasible motion space in Section V. This extension enhances the adaptability of the motion planning strategy by providing a more accurate approximation of the safe region.

\section{Data-Driven Motion Planning Using Invariant Ellipsoids}


This section presents a data-driven framework for safe motion planning that relies on invariant ellipsoids to ensure safety guarantees and dynamic feasibility at every step of the planning process. Unlike traditional sampling-based methods, such as Rapidly-exploring Random Trees (RRT), which focus primarily on kinematic feasibility and ignore system dynamics, the proposed approach directly integrates data-driven control synthesis with path planning.

Specifically, our approach can be viewed as a safe and dynamics-aware extension of RRT tailored for uncertain nonlinear systems that rely on collected data for planning. By constructing invariant ellipsoids around each sampled point in the 2D position space \( \mathcal{X}_{pos} \), the method ensures that the full system state \( x_k \in \mathbb{R}^n \) remains within a certified safe region, while the position state \( x_{pos,k} \) transitions between waypoints. The framework is divided into several key stages: identifying the convex admissible set, ensuring safety guarantees using single ellipsoids, designing safe transitions between successive ellipsoids, and executing the planned trajectory using the computed control gains. Each step is outlined in detail in the following subsections.

\subsection{Identifying the Convex Admissible Set}
In this subsection, we focus on identifying the admissible set for each sampled point \( p_{s} =[x_s,y_s]^T \in \mathbb{R}^2 \). The overall admissible set is typically non-convex due to environmental constraints and the presence of obstacles. To facilitate motion planning, we partition this non‑convex admissible set into multiple polytopes. Each polytope is described by the linear‑inequality set $\mathcal{X}_{\text{pos},\kappa}(F_{\text{pos}},g_{\text{pos}})$, which specifies the admissible‑region boundaries.

Given a sampled point \( p_{s} \), we determine which convex polytope it belongs to and update the corresponding polyhedral constraints. If multiple polytopes contain the sampled point, a random selection is made to ensure flexibility in planning. This step is critical as it forms the basis for constructing the invariant ellipsoids used in the subsequent stages.

The steps for identifying the admissible set are summarized in Algorithm 1.

\begin{algorithm}
\caption{Identifying the Admissible Set}
\label{alg:admissible_set}
\begin{algorithmic}[1]
\Require Sampled point $p_{s} = [x_s,y_s]^T$, obstacle dimension $\mathcal{O}$, environment bounds $[x_{s,\min}, x_{s,\max}, y_{s,\min}, y_{s,\max}]$.
\Ensure Polyhedral safe region $\mathcal{X}_{pos,i}(F_{pos}, g_{pos})$ around $p_{s}$.

\State \textbf{Determine} the location of $p_{s}$ relative to $\mathcal{O}$.
\State \textbf{Identify} candidate polytopes $\mathcal{X}_{pos,i}(F_{pos}, g_{pos})$ that contain $p_{s}$ for $i=1,\ldots,\nu$.
\If{multiple polytopes are found}
    \State Select a polytope $\mathcal{X}_{pos,i,\rho}(F_{pos}, g_{pos})$ randomly, where $\rho \sim \{1, \dots, \iota\}$.
    \State Assign boundaries of $\mathcal{X}_{pos,i,\rho}(F_{pos}, g_{pos})$ to $p_{s}$.
\Else
    \State Assign the boundaries of the identified polytope to $p_{s}$.
\EndIf
\State \textbf{Select} the polyhedral constraints $F_{\text{pos}}p_s \le g_{\text{pos}}$ associated with the chosen polytope.
\State \Return $(F_{\text{pos}}, g_{\text{pos}})$.
\end{algorithmic}
\end{algorithm}
The symbols \( \nu \) and \( \iota \) represent the total number of constraint sets that define the non-convex state space and the number of constraint sets that contain a given sampled location, respectively. The matrices \( F_{pos} \) and \( g_{pos} \) define the polyhedral constraint set in the 2D position space. These constraints are then lifted to form the full-state polyhedral set \( \mathcal{F} \zeta_e \leq \bar{g} \), which incorporates constraints on other states. This full-state constraint set is subsequently used for invariant ellipsoid design and control synthesis.

\begin{remark}\label{remark_1}
Although the admissible region \( \mathcal{X} \subseteq \mathbb{R}^n \) is defined over the full system state, we identify a convex admissible subset in the 2D position space, denoted by \( \mathcal{X}_{\mathrm{pos}} \subseteq \mathbb{R}^2 \), for each sampled point. This region accounts for environmental constraints such as obstacles and workspace boundaries. Once \( \mathcal{X}_{\mathrm{pos}} \) is identified, it is used to define or constrain the corresponding full-state admissible set \( \mathcal{X} \), within which an invariant ellipsoid is constructed to guarantee safety and feasibility.
\end{remark}

Once the admissible convex set for the sampled point $p_s$ is identified, the next step is to compute a feedback gain that ensures the largest possible subset of this set remains invariant. This guarantees that the system stays within the safe region at all times, providing critical safety guarantees during motion planning.

\subsection{Dynamics-Aware Safety Guarantees Using Single Ellipsoids}
In this subsection, we aim to design a data-driven state-feedback controller that ensures the full system state remains within the largest invariant ellipsoid, constructed inside the full-state admissible set corresponding to a convex region identified in the 2D position space. This step involves computing the gain that guarantees contractiveness, allowing the system to safely evolve while respecting the physical constraints imposed by the admissible set.

Before formulating the problem, it is essential to highlight the importance of contractive sets in maintaining safety. Contractive sets serve as a foundation for keeping the system’s state within specified boundaries over time, which is crucial for safety-critical applications. This framework simplifies the design of controllers that can enforce these boundaries effectively.

\begin{definition}[Contractive Set \cite{modares2023}] \label{def_4} 
A set $\mathcal{S} \subseteq \mathbb{R}^{n+m+2}$ is called $\lambda$-contractive for the system \eqref{eq.final_error_system} if, for all $k \geq 0$, $\zeta_{e,k} \in \mathcal{S}$ implies that $\zeta_{e,k+1} \in \lambda\mathcal{S}$, where $0 < \lambda \leq 1$. 
\end{definition}

The problem of data-driven safe control is formally stated below, focusing on maintaining the system’s state within predefined safe regions while ensuring compliance with system constraints.

\begin{problem}[Data-Driven Safe Control Using Single Ellipsoids]\label{Problem_1}
Consider the nonlinear error system described in \eqref{eq.final_error_system} under Assumptions 1–5. The admissible set is given by $\mathcal{S}(\mathcal{F},\bar{g})$, and consider a safe set represented by the ellipsoid $\mathcal{E}(\mathcal{P}, 0)$. Design a data-driven nonlinear state-feedback controller as 
\begin{equation}\label{nonlinear_controller}
v_k = \mathcal{K} \mathcal{Z}(\xi_k, \eta_k) = \begin{bmatrix} \mathcal{K}_l & \bar{\mathcal{K}}_l & \mathcal{K}_{nl} \end{bmatrix} \begin{bmatrix} \begin{bmatrix} x_k \\ u_k \end{bmatrix} \\ \eta_k \\ S(\xi_k) \end{bmatrix},
\end{equation}
to maximize the size of the invariant ellipsoid $\mathcal{E}(\mathcal{P}, 0) \subseteq \mathcal{S}$.
\end{problem}

To address this problem, we utilize the data-driven representation \eqref{eq.data_driven_error_system} and the collected data \eqref{eq.data_V0}--\eqref{eq.data_Z0}. The control design involves two types of gains: the nonlinear gain \( \mathcal{K}_{nl} \), which is optimized to reduce the upper bound on nonlinear residuals, and the linear gains \( \mathcal{K}_{l} \) and \( \bar{\mathcal{K}}_{l} \), which are designed to ensure that the largest possible set remains \( \lambda \)-contractive, thereby maintaining safety even in the presence of residual nonlinearities.

\begin{theorem}[Safe Control Design with Single Ellipsoids]\label{thm:single_ellipsoid}
Consider the nonlinear system \eqref{eq.final_error_system} that satisfies Assumptions 1--5 . Data are collected and arranged as equations \eqref{eq.data_V0}--\eqref{eq.data_Z0}. Let there exist matrices $\mathcal{P} \in$ $\mathbb{S}^{n+m+2}$ and $\mathcal{Y} \succeq 0$, and positive scalar $\gamma$. Consider any feasible solution of the following optimization problem
\begin{align}
& \min_{\mathcal{P}, \mathcal{Y}, \gamma, \mathcal{G}_{K,nl}} \quad \gamma - \operatorname{logdet}(\mathcal{P}) \label{eq.SDP_single} \\
& \mathrm{s. t.} \nonumber \\
& \mathcal{Z}_0\mathcal{Y} = \begin{bmatrix} \mathcal{P} \\ 0_{(n_{\mathcal{z}}-n-m-2) \times (n+m+2)} \end{bmatrix}, \label{eq.DD_CE_conditions_1} \\
& \mathcal{Z}_0\mathcal{G}_{K,nl} = \begin{bmatrix} 0_{n+m+2 \times (n_{\mathcal{z}}-n-m-2)} \\ I_{(n_{\mathcal{z}}-n-m-2)} \end{bmatrix}, \label{eq.DD_CE_conditions_2} \\
& \begin{bmatrix}
\mathcal{P} & \sqrt{1+\tau}\Xi_1 \mathcal{Y} \\
(*) & \lambda \mathcal{P}-\varepsilon(1+\tau^{-1}) \mathcal{P}
\end{bmatrix} \succeq 0, \label{eq.DD_CE_conditions_3} \\
& \begin{bmatrix}
\gamma I_{(n+m+2)} & \Xi_1 \mathcal{G}_{K,nl} \\
(*) & \gamma I_{(n_{\mathcal{z}}-n-m-2)}
\end{bmatrix} \succeq 0, \label{eq.DD_CE_conditions_4} \\
& \begin{bmatrix}
\varepsilon I_{(n+m+2)} & I_{(n+m+2)} \\
(*) & \mathcal{P}
\end{bmatrix} \succeq 0, \label{eq.DD_CE_conditions_5} \\
& \begin{bmatrix}
I_{(n+m+2)} & \gamma \mathcal{Q} \\
(*) & \mathcal{P}
\end{bmatrix} \succeq 0, \,\,\,
\begin{bmatrix}
\mathcal{P} & \mathcal{P}\mathcal{F}_l^T \\
(*) & \bar{g}_l^2
\end{bmatrix} \succeq 0, \,\,\, \forall l=1,\ldots,q. \label{eq.DD_CE_conditions_7}
\end{align}

\noindent where \( \mathcal{Q} \) is the extended form of the Lipschitz matrix \( Q \) from Assumption~1, applied to the error dynamics. Then, for some $\tau>0$ and $\varepsilon \geq \lambda_{max}(\mathcal{P}^{-1})$, Problem 2 is solved and nonlinear state-feedback gain is computed as $\mathcal{K}=[\mathcal{K}_l \quad \mathcal{K}_{nl}]=V_0[\mathcal{G}_{K,l} \quad \mathcal{G}_{K,nl}]$ where $\mathcal{G}_{K,l}=\mathcal{Y}\mathcal{P}^{-1}$.
\end{theorem}

\begin{proof}
See Appendix~\ref{app:proof_single_ellipsoid}.
\end{proof}

\begin{corollary}
According to \eqref{eq.proof_CE_6}, if the system state starts in the ellipsoid defined by \( \mathcal{P} \), the next state will lie within the scaled ellipsoid \( \lambda' \mathcal{P} \), where
\begin{equation}
\lambda' = \frac{\lambda - \varepsilon(1 + \tau^{-1})}{1 + \tau}.
\end{equation}
\end{corollary}

Since \(\lambda' < \lambda\), this contraction mitigates the effect of nonlinearities by providing a margin of safety, reducing the influence of these nonlinearities, and enhancing the robustness and stability of the closed-loop system.

\begin{remark}\label{remark_2}
The invariant ellipsoids obtained from Theorem~1 are centered at the origin in the error coordinate system, which corresponds to the desired steady-state \( \zeta^\ast \). When these ellipsoids are projected onto the 2D position subspace and translated to be centered at the sampled waypoint \( x_{\mathrm{pos}} \), they define the safe region for that waypoint in the position space. Therefore, each 2D-projected and translated ellipsoid represents a safe set within which the position state can evolve, while guaranteeing that the full system state remains inside the corresponding invariant set.
\end{remark}

For each sampled position, an invariant ellipsoid centered at the desired steady state is computed using Theorem~2. However, not all sampled points can be selected as new waypoints. To ensure safety and feasibility, a candidate waypoint is accepted only if a safe transition exists from a previously verified node—specifically, if the corresponding invariant ellipsoids overlap. In the next subsection, we address this critical aspect by developing a method to certify safe transitions between successive invariant ellipsoids, which guarantees system invariance throughout the planned path.

\subsection{Safe Transition Between Invariant Ellipsoids}
For successful motion planning, it is essential to ensure smooth transitions between successive ellipsoids. A major limitation of traditional invariant‑ellipsoid methods is the restrictive assumption that the center of the parent ellipsoid must lie inside the child ellipsoid \cite{niknejad2024soda}. In those schemes the child ellipsoid’s controller is applied immediately after the switch, so its gain must already stabilize the state starting from the parent center; embedding that center in the child ellipsoid guarantees the transition is always feasible. This requirement severely narrows the set of admissible trajectories, often yielding sub‑optimal or even infeasible plans. To remove this drawback, we relax the center‑containment condition and instead permit successive ellipsoids to intersect, certifying safety through an intermediate invariant ellipsoid.

By solving a set of linear matrix inequalities (LMIs), the next ellipsoid is constructed to intersect with its parent, thereby relaxing the conservative requirement of containing the parent ellipsoid's center while still guaranteeing a safe transition. Since the intersection of the ellipsoids is only crucial in the \(x\)-\(y\) plane, the ellipsoids are projected onto this plane before computing the intersection. An intermediate ellipsoid is then identified that not only contains this intersection but also remains within the union of both projected ellipsoids in the \(x\)-\(y\) plane. This guarantees that the system state always evolves within a well-defined safe region.

The key advantage of this formulation is that the center of the computed intermediate ellipsoid lies within both the parent and successor ellipsoids in the \(x\)-\(y\) plane. This property facilitates a smooth transition between control policies. Initially, the parent ellipsoid’s gain is employed to steer the system state toward the intermediate node. Once the state reaches this region, the control switches to the successor ellipsoid’s gain, guiding the trajectory toward the next waypoint.

Given two ellipsoids
\begin{equation}
\mathcal{E}_i = \left\{ x_{pos} \mid (x_{pos} - p_{s,i})^T \mathcal{P}_{proj,i}^{-1} (x_{pos} - p_{s,i}) \leq 1 \right\}, \quad i=1,2,
\end{equation}
with $p_{s,i}=[x_i,y_i]^\top$ denotes the sampled point's position for the $i$th ellipsoid for $i=1,2$, We aim to determine an intermediate ellipsoid
\begin{equation}
\mathcal{E}_o = \left\{ x_{pos} \;\middle|\; (x_{pos} - p_{s,o})^\top \mathcal{P}_o^{-1} (x_{pos} - p_{s,o}) \leq 1 \right\}
\end{equation}
that satisfies the following conditions:
\begin{enumerate}
    \item \( \mathcal{E}_o \) encloses the intersection of \( \mathcal{E}_i \) in the \( x \)-\( y \) plane.
    \item \( \mathcal{E}_o \) is contained within the union of \( \mathcal{E}_i \) in the \( x \)-\( y \) plane.
\end{enumerate}

Following the SDP formulation in \cite{wang2019equivalence}, we solve
\begin{align}\label{eq.intersection_SDP}
& \max_{\varrho_i} \;\;\; \operatorname{logdet} \left( \left(\sum_{i=1}^{2} \varrho_i \mathcal{P}_{proj,i}^{-1} \right)^{-1} \right) \\
& \mathrm{s. t.} \nonumber \\
& \begin{bmatrix}
1 - \sum\limits_{i=1}^{2} \varrho_i + \sum\limits_{i=1}^{2} \varrho_i p_{s,i}^T \mathcal{P}_{proj,i}^{-1} p_{s,i}  & \sum\limits_{i=1}^{2} \varrho_i p_{s,i}^T \mathcal{P}_{proj,i}^{-1}  \\
\sum\limits_{i=1}^{2} \varrho_i \mathcal{P}_{proj,i}^{-1} p_{s,i} & \sum\limits_{i=1}^{2} \varrho_i \mathcal{P}_{proj,i}^{-1}
\end{bmatrix} \succeq 0, \\
& \varrho_i \geq 0, \quad i=1,2.
\end{align}

If the above problem is feasible, then an optimal solution for the intersection ellipsoid exists, and its shape matrix and center in the $x$-$y$ plane are given by
\begin{equation}
\mathcal{P}_{proj,o}^{-1} = \sum\limits_{i=1}^{2} \varrho_i^* \mathcal{P}_{proj,i}^{-1},
\end{equation}
\begin{equation}
p_{s,o} = \mathcal{P}_{proj,o} \sum\limits_{i=1}^{2} \varrho_i^* \mathcal{P}_{proj,i}^{-1} p_{s,i}.
\end{equation}
where $\varrho_i^*$ are the optimal values of $\varrho_i$.

\begin{figure}[h]\label{fig:safe_transition}
    \centering
    \includegraphics[width=0.7\linewidth]{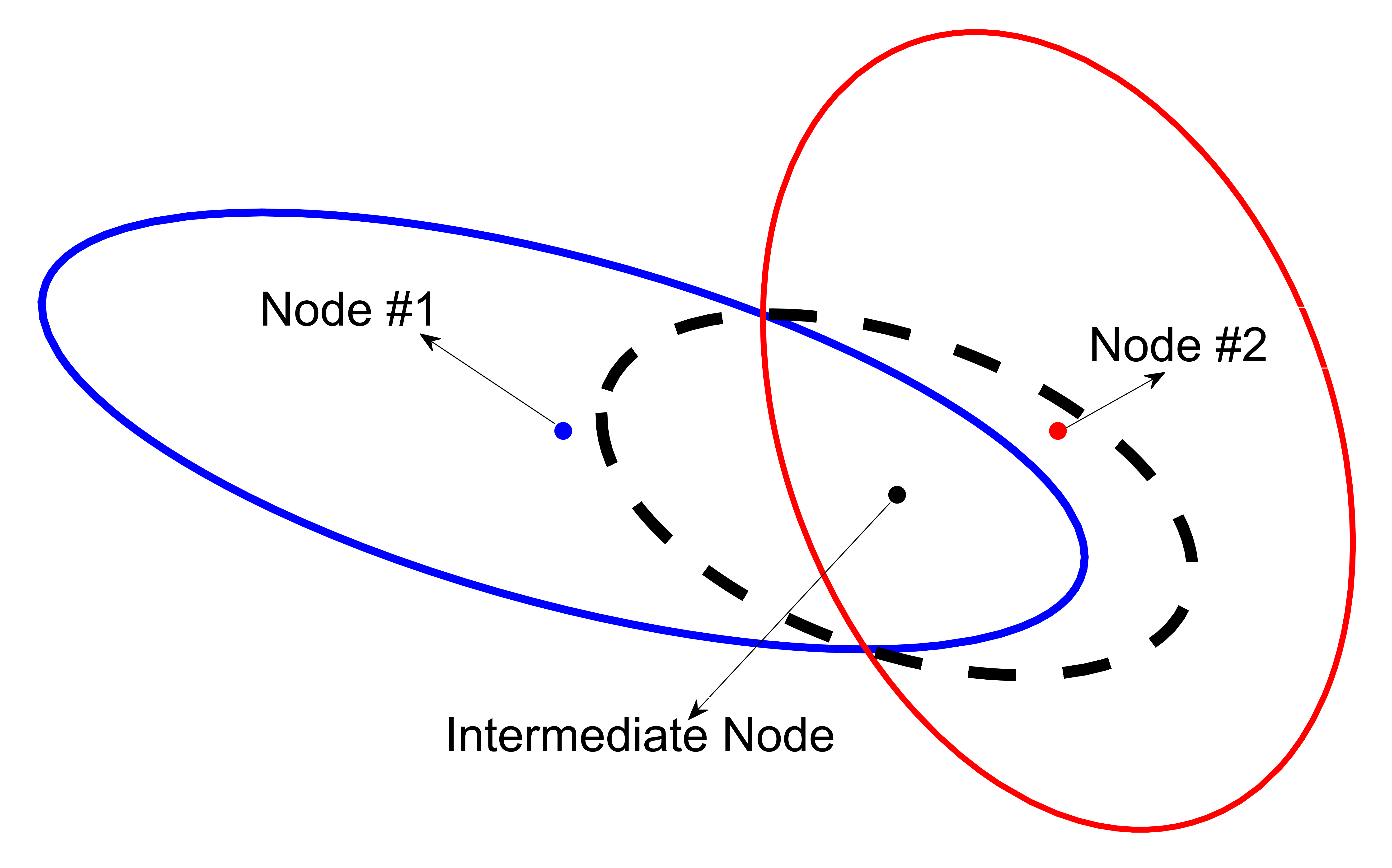}
    \caption{Safe transition between two ellipsoidal regions. The blue ellipsoid represents the parent safe region, while the red ellipsoid represents the child region. The dashed black ellipsoid is the intermediate invariant ellipsoid ensuring a smooth and safe transition. The system state moves from Node \#1 to the Intermediate Node and then follows the next controller to Node \#2.}
    \label{fig:safe_transition}
\end{figure}

Figure~\ref{fig:safe_transition} illustrates this transition mechanism. The blue ellipsoid represents the parent safe set, and the red ellipsoid represents the child region. Instead of enforcing complete containment of the parent ellipsoid's center within the child, an intermediate ellipsoid (dashed black) is computed. This intermediate ellipsoid contains the intersection of the two ellipsoids while remaining within their union.

The state initially moves from the parent node (blue dot) to the intermediate node (black dot) using the control policy associated with the parent ellipsoid. Once the state reaches the intermediate node, the control policy switches to the gain associated with the child ellipsoid, guiding the system toward the child node (red dot). This method eliminates the conservativeness associated with conventional approaches while maintaining safety guarantees. By ensuring that transitions between ellipsoidal regions remain dynamically feasible, the proposed approach extends the applicability of invariant ellipsoids to complex motion planning problems.

Algorithm 2 summarizes the steps for ensuring a safe transition between a parent ellipsoid and its corresponding child ellipsoid. 

\begin{algorithm}\label{alg:safe_transition}
\caption{Safe Transition Between Projected Invariant Ellipsoids}
\label{alg:ellipsoid_safe_transition}
\textbf{Input:} Two invariant ellipsoids $\mathcal{E}_1$ and $\mathcal{E}_2$.\\
\textbf{Output:} Intermediate ellipsoid $\mathcal{E}_o$ ensuring safe transition in the $x$-$y$ plane.

\begin{algorithmic}[1]
\State \textbf{Project} the shape matrices $\mathcal{P}_1$ and $\mathcal{P}_2$ onto the $x$-$y$ plane to obtain $\mathcal{P}_{proj,1}$ and $\mathcal{P}_{proj,2}$.
\State \textbf{Formulate} the semidefinite program (SDP) \eqref{eq.intersection_SDP} to find an intermediate ellipsoid $\mathcal{E}_o$ in the $x$-$y$ plane such that:
\begin{itemize}
    \item $\mathcal{E}_o$ encloses the intersection of $\mathcal{E}_{proj,1}$ and $\mathcal{E}_{proj,2}$.
    \item $\mathcal{E}_o$ is contained within the union of $\mathcal{E}_{proj,1}$ and $\mathcal{E}_{proj,2}$.
\end{itemize}
\State \textbf{Solve} the SDP to determine the shape matrix $\mathcal{P}_{proj,o}$ and center $p_{s,o}$ of the intermediate ellipsoid in the $x$-$y$ plane.
\State \textbf{Use} the parent ellipsoid’s gain to steer the system state toward the center of $\mathcal{E}_o$.
\State \textbf{Switch} to the successor ellipsoid’s gain to guide the state from $\mathcal{E}_o$ to the next node.
\State \textbf{Return} the intermediate ellipsoid $\mathcal{E}_o$ defined by $(\mathcal{P}_{proj,o}, p_{s,o})$.
\end{algorithmic}
\end{algorithm}

Before implementing the control policies described in Theorem 2, it is crucial to define how the planned path will be executed while ensuring system safety.

\subsection{Execution of Safe Motion Planning} 
Once the path is planned, the next step is to safely execute the trajectory while ensuring that the system state remains within the sequence of invariant ellipsoids. In this framework, each node in the path corresponds to an invariant ellipsoid with a precomputed control gain that guarantees that the system state converges to the center of the target ellipsoid without leaving its boundaries.

During execution, the system uses the parent ellipsoid's control gain to reach the child ellipsoid, after which it switches to the child's gain for continued motion.

The steps for this process are summarized in Algorithm~\ref{alg:execution}, which ensures that the trajectory is executed safely and the system transitions smoothly between ellipsoidal regions.



\begin{algorithm}
\caption{Execution of Safe Motion Plan via Invariant Ellipsoids}
\label{alg:execution}
\textbf{Input:} Precomputed waypoint sequence \( \{p_1, p_2, \dots, p_{N_w}\} \) with corresponding control gains \( \{\mathcal{K}_1, \mathcal{K}_2, \dots, \mathcal{K}_n\} \), system dynamics \( f(x, u) \), initial position \( p_{init} = p_1 = x_{pos,0} \).\\
\textbf{Output:} Safe trajectory execution from \( p_{init} \) to \( p_{goal} = p_{N_w} \).

\begin{algorithmic}[1]
\State \textbf{Initialize} current waypoint index \( i \gets 1 \), current position \( p_i \gets p_{init} \)
\While{$p_{curr} \neq p_{N_w}$}
    \State \textbf{Retrieve} the control gain \( \mathcal{K}_i \) associated with waypoint \( p_i \)
    \State \textbf{Compute} virtual control \( v_k = \mathcal{K}_i \zeta_k \), and update the actual control input \( u_k \) as the integral of \( v_k \).
    \State \textbf{Update} state using system dynamics: \( x_{k+1} = f(x_k, u_k) \)
    \State \textbf{Update} position \( p_{curr} \gets \text{position component of } x_{k+1} \)
    \If{$\|p_{curr} - p_i\| < r_f$ \textbf{and} \( i < N_w \)}
        \State \textbf{Advance} to the next waypoint: \( i \gets i + 1 \)
    \EndIf
\EndWhile
\State \textbf{Return} executed trajectory
\end{algorithmic}
\end{algorithm}

Here, \( p_{curr} \) denotes the position component of the full system state \( x_k \), which is used to determine the current location of the system within the sequence of ellipsoids and to track progress toward the next waypoint.

\subsection{Summary of the Data-Driven Safe Motion Planning Method}
The proposed data-driven safe motion planning framework leverages invariant ellipsoids to ensure system safety and feasibility at each step of the planning process. This method focuses on constructing a safe trajectory from the initial waypoint $p_{init}$ to the goal waypoint $p_{goal}$, while satisfying safety guarantees. The primary steps of the framework are as follows:

\begin{itemize}
    \item \textbf{Initialization:} The algorithm begins by initializing the graph $G$ with the initial waypoint $p_{init}$. The corresponding invariant ellipsoid $\mathcal{E}_{init}$ is computed using data-driven LMIs to ensure that the initial waypoint lies within a safe region.
    
    \item \textbf{Sampling and Admissible Set Identification:} At each iteration, a new point \( p_s \) is sampled with a specified probability of selecting a point near the goal. The admissible set around the sampled point is identified.

    \item \textbf{Ellipsoid and Control Gain Computation:} For each sampled point, a set of LMIs is solved to determine an invariant ellipsoid \( \mathcal{E}_{s} \) that guarantees safety. Simultaneously, a nonlinear state-feedback control gain \( \mathcal{K}_{s} \) is computed to regulate the system within the ellipsoidal region.

    \item \textbf{Intersection Check and Graph Update:} The algorithm checks for an intersection between the invariant ellipsoids \( \mathcal{E}_{\text{nearest}} \) and \( \mathcal{E}_{s} \) by solving an SDP. If an intersection exists, an intermediate point is computed within the intersection region to facilitate a smooth transition. Both the sampled node \( p_{s} \) and the intermediate node are then added to the graph \( G \). Otherwise, if no intersection is found, the sampled point is rejected, and a new point is sampled. This process is repeated iteratively to construct a feasible trajectory from \( p_{init} \) to \( p_{goal} \), ensuring smooth and safe motion while maintaining invariance within the ellipsoidal regions.

    \item \textbf{Trajectory Execution with Control Gains:} Once a valid path is constructed, the planned trajectory is executed using the system dynamics and the control gains associated with each invariant ellipsoid along the path. The system first applies the control gain \( \mathcal{K}_{\text{nearest}} \) associated with the nearest ellipsoid \( \mathcal{E}_{\text{nearest}} \) to steer toward the intermediate point. Upon reaching this region, the control switches to the gain \( \mathcal{K}_{s} \) associated with the sampled ellipsoid \( \mathcal{E}_{s} \), guiding the system toward its center. This process is repeated iteratively, following the planned sequence of ellipsoids, until the system reaches the goal point, ensuring smooth transitions and stability while respecting the safety constraints imposed by the invariant ellipsoids.

\end{itemize}

The detailed steps of the method are summarized in Algorithm~\ref{alg:IE-RRT}.
\begin{algorithm}
\caption{Data-Driven Safe Motion Planning Using Invariant Ellipsoids}
\label{alg:IE-RRT}
\begin{algorithmic}[1]
\Require $p_{init}$, $p_{goal}$, collected data $V_0$, $\Xi_1$, and $\mathcal{Z}_0$, contraction factor $\lambda$, maximum iterations $N_{\text{max}}$, probability $p_r$ of goal sampling.
\Ensure A safe dynamically feasible path from $p_{init}$ to $p_{goal}$ using invariant ellipsoids.

\State \textbf{Initialize} graph $G = (V, E)$ with $V = \{p_{init}\}$, $E = \emptyset$.
\State \textbf{Compute} the initial invariant ellipsoid \( \mathcal{E}_{init} \) and its corresponding nonlinear state-feedback gain around \( p_{init} \) using the SDP formulation in \eqref{eq.SDP_single}.
\While{$\|p_s - p_{goal}\| \geq r_f$ \textbf{and} iteration $\leq N_{\text{max}}$}
    \State \textbf{Sample} a new point $p_{s}$ using probability $p_r$.
    \State \textbf{Identify} the admissible set around $p_{s}$ and compute the corresponding invariant ellipsoid $\mathcal{E}_{s}$ and its corresponding nonlinear state-feedback gain.
    \State \textbf{Check} for safe transition by verifying the intersection between $\mathcal{E}_{nearest}$ and $\mathcal{E}_{s}$.
    
    \If{an intersection exists}
        \State \textbf{Compute} an intermediate invariant ellipsoid $\mathcal{E}_{\text{int}}$ that contains the intersection using SDP \eqref{eq.intersection_SDP}.
        \State \textbf{Assign} the control gain \( \mathcal{K}_{\text{nearest}} \) to steer the system toward the intermediate node. Upon reaching this region, switch to the control gain \( \mathcal{K}_{s} \) to guide the system toward the center of \( \mathcal{E}_{s} \).
        \State \textbf{Add} $p_{s}$ and the intermediate node to $G$.
    \EndIf
\EndWhile
\State \textbf{Execute} the planned trajectory using system dynamics and associated control gains for each ellipsoid.
\State \Return executed trajectory
\end{algorithmic}
\end{algorithm}

\section{Data-Driven Motion Planning Using Convex Hull of Ellipsoids}
In this section, we extend the data-driven motion planning framework by leveraging convex hulls of ellipsoids instead of individual invariant ellipsoids. The key advantage of using convex hulls is their ability to represent larger and more complex safe regions, which provides greater flexibility in motion planning. Unlike single ellipsoids, convex hulls offer improved coverage of admissible sets and allow for smoother transitions between regions without sacrificing safety guarantees. By constructing the convex hull of multiple ellipsoids, the proposed method ensures that the state remains within a larger, well-defined safe region at each step. This approach enables the system to explore more efficient trajectories. 

The steps of the method—checking safe transitions, and executing the planned trajectory—are adapted to accommodate the convex hull representation. The process of identifying admissible sets remains the same as for individual ellipsoids, focusing on partitioning the non-convex set into convex subsets for subsequent computations.

\subsection{Dynamics-Aware Safety Guarantees Using the Convex Hull of Ellipsoids}
\begin{problem}[Data-Driven Safe Control Using the Convex Hull of Ellipsoids]\label{Problem_3} 
Consider the nonlinear system \eqref{eq.final_error_system} under Assumptions 1--5, along with the admissible polyhedral set \( \mathcal{S} \). The objective is to design a partitioning scheme \( \mathcal{C}_1,\ldots,\mathcal{C}_{n_p} \) and a piecewise-affine state-feedback controller of the form

\begin{equation}\label{eq.general_controller}
v_k=
\left\{\begin{array}{cc}
\mathcal{K}_1^p \mathcal{Z}(\xi_k,\eta_k), \quad \text{if} \quad \zeta_{e,k} \in \mathcal{C}_1 \\
\vdots \\
\mathcal{K}_{n_p}^p \mathcal{Z}(\xi_k,\eta_k), \quad \text{if} \quad \zeta_{e,k} \in \mathcal{C}_{n_p}
\end{array}\right.
\end{equation}

such that the controlled invariant set \( \mathcal{S}_c = \bigcup_{i=1}^{n_p} \mathcal{C}_i \) is maximized while ensuring \( \mathcal{S}_c \subseteq \mathcal{S} \) and remains invariant for the closed-loop system. Here, \( n_p \) denotes the number of partitions formed within the convex hull of ellipsoids.
\end{problem}

\begin{theorem}[Safe Control Design Using the Convex Hull of Ellipsoids]\label{thm:CHE} 
Consider system \eqref{eq.final_error_system} that satisfies Assumptions 1--5. Data are collected and arranged as equations \eqref{eq.data_V0}--\eqref{eq.data_Z0}. Let there exist matrices $\mathcal{P}_i \in \mathbb{S}^{(n+m+2)}$ and $\mathcal{Y}_i \succeq 0$, and positive scalars $\vartheta_i$ for $i=1,\ldots,n_e$. Consider any feasible solution of the following optimization problem
\begin{align}\label{eq.SDP_CHE}
& \min_{\mathcal{P}_i, \mathcal{Y}_i, \gamma_i, \vartheta_i, \mathcal{G}_{K,nl,i}} \quad \sum_{i=1,\ldots,n_e} \gamma_i - \vartheta_i \\
& \mathrm{s. t.} \nonumber \\
& \mathcal{Z}_0 \mathcal{Y}_i = \begin{bmatrix} \mathcal{P}_i \\ 0_{(n_\mathcal{z}-(n+m+2)) \times (n+m+2)} \end{bmatrix}, \label{eq.DD_CE_conditions_1_CHE} \\
& \mathcal{Z}_0 \mathcal{G}_{K,nl,i} = \begin{bmatrix} 0_{(n+m+2) \times (n_\mathcal{z}-(n+m+2))} \\ I_{(n_\mathcal{z}-(n+m+2))} \end{bmatrix}, \label{eq.DD_CE_conditions_2_CHE} \\
& \begin{bmatrix}
\mathcal{P}_i & \Xi_1 \mathcal{Y}_j \\
(*) & \lambda_j' \mathcal{P}_j
\end{bmatrix} \succeq 0, \label{eq.DD_CE_conditions_3_CHE} \\ 
& \begin{bmatrix}
\gamma_i I_{(n+m+2)} & \Xi_1 \mathcal{G}_{K,nl,i} \\
(*) & \gamma_i I_{(n_\mathcal{z}-(n+m+2))}
\end{bmatrix} \succeq 0, \label{eq.DD_CE_conditions_4_CHE} \\ 
& \begin{bmatrix}
\varepsilon_i I_{(n+m+2)} & I_{(n+m+2)} \\
(*) & \mathcal{P}_i
\end{bmatrix} \succeq 0, \,\,\,
\begin{bmatrix}
I_{(n+m+2)} & \gamma_i \mathcal{Q} \\
(*) & \mathcal{P}_i
\end{bmatrix} \succeq 0, \label{eq.DD_CE_conditions_6_CHE} \\
& \begin{bmatrix}\label{eq.DD_CE_conditions_7_CHE}
\mathcal{P}_i & \mathcal{P}_i\mathcal{F}_l^T \\
(*) & \bar{g}_l^2
\end{bmatrix} \succeq 0, \,\,\, \forall l=1,\ldots,q, \\
& \begin{bmatrix}\label{eq.DD_CE_conditions_8_CHE}
1 & \vartheta_i d_i^T \\
(*) & \mathcal{P}_i
\end{bmatrix} \succeq 0, \,\,\,
\mathrm{for} \; i=1,\ldots,n_e,
\end{align}
where $\lambda_j' = \frac{\lambda - \varepsilon_j(1 + \tau_j^{-1})}{1 + \tau_j}$. Then Problem \ref{Problem_3} is solved, and the largest invariant subset of the admissible set~$\mathcal{S}$ that can be represented as the convex hull of the ellipsoids computed above for the closed‑loop system~\eqref{eq.final_error_system} is given by \( \mathcal{S}_c = \operatorname{Co} \big( \mathcal{E}(\mathcal{P}_1,0),\ldots,\mathcal{E}(\mathcal{P}_{n},0) \big) \). The corresponding controller gains are computed as \( \mathcal{K}_i = [\mathcal{K}_{l,i} \quad \mathcal{K}_{nl,i}] = V_0 [\mathcal{G}_{K,l,i} \quad \mathcal{G}_{K,nl}] \), where \( \mathcal{G}_{K,l,i} = \mathcal{Y}_i \mathcal{P}_i^{-1} \). Additionally, the index \( j \) is defined as \( j = \operatorname{R}_{n_e}(i) = \operatorname{mod}(i + n_e - 2, n_e) + 1 \) for \( i = 1, \ldots, n_e \), where \( n_e \) denotes the number of ellipsoids.
\end{theorem}

\begin{proof}
See Appendix~\ref{proof_thm_3}.
\end{proof}

\begin{remark}
Although Theorem~\ref{thm:CHE} provides the procedure for computing the individual ellipsoids \(\mathcal{E}(\mathcal{P}_i,0)\), their feedback gains \(\mathcal{K}_i\), and proves that the resulting convex hull is \(\lambda\)-contractive, the systematic partitioning of the admissible set, the assignment of a gain to each partition, and the construction of the overall control law from the collection of gains \(\{\mathcal{K}_i\}\) are developed in detail later and summarized in Algorithms 5 and 6.
\end{remark}

\begin{remark}\label{remark_3}
In general, all possible non-adjacent vertex pairs can be used to define reference directions for constructing ellipsoids. These directions correspond to diagonals of the polytope formed by system constraints and offer rich geometric coverage. However, as the number of system states and vertices increases, the number of such directions grows quadratically, significantly increasing computational complexity. To balance computational efficiency with coverage, we consider only a subset of directions aligned with the principal axes. As a result, for a system with \( n \) states, we construct \(n_e = n \) ellipsoids—each aligned with one of the principal directions—thus offering a scalable and tractable solution for high-dimensional systems.
\end{remark}

\subsubsection{Partitioning and State-Feedback Control Computation}
This subsection describes the process of partitioning the convex hull of ellipsoids and computing state-feedback controllers for each partition. The partitioning extends the approach in \cite{nguyen2024} to higher-dimensional systems using an algorithmic approach.

\begin{definition}
A point $v^*$ on the boundary of the convex set $\mathcal{S}$, denoted as $\operatorname{Fr}(\mathcal{S})$, is called an extreme point if it cannot be expressed as a convex combination of any other points in $\mathcal{S}$.
\end{definition}

The vertices of the convex hull are obtained by solving:
\begin{equation}\label{eq.vertices}
v_c^T \mathcal{P}_i v_c = 1, \quad i=1,\dots,n_e.
\end{equation}

Not all solutions correspond to true vertices of the convex hull; only those forming the outer boundary are retained. The partitioning method is summarized in Algorithm \ref{alg.set}.

\begin{algorithm}
\caption{Set Partitioning Algorithm}\label{alg.set}
\begin{algorithmic}[1]
\Require $\mathcal{P}_i$: Ellipsoid shape matrices; $2n+m$: Number of augmented states.
\Ensure $v_e^*$: Vertices forming the convex hull.

\State \textbf{for} each pair of ellipsoids $(\mathcal{E}(\mathcal{P}_i,0), \mathcal{E}(\mathcal{P}_j,0))$ \textbf{do} 
\State \quad Solve:
\begin{align}
\phi^T \mathcal{P}_i \phi &= 1 \nonumber \\
\phi^T \mathcal{P}_j \phi &= 1 \nonumber
\end{align}
\State \quad Compute candidate vertices:
\begin{align}
v_{e,i} &= \mathcal{P}_i \phi \nonumber \\
v_{e,j} &= \mathcal{P}_j \phi \nonumber
\end{align}
\State \quad Store all candidate vertices:
\begin{equation}
v_{\text{all}} = [v_{\text{all}}, v_{e,i}, v_{e,j}] \nonumber
\end{equation} 
\State \textbf{end for} 

\State Apply the Quickhull algorithm~\cite{barber1996quickhull} to compute \( v_e^* \) and their corresponding convex hull.
\State \textbf{Partitioning:} Each region is formed by selecting local neighborhoods of adjacent extreme points to define a polyhedral partition.
\end{algorithmic}
\end{algorithm}

\subsubsection{Numerical Example: Partitioning the Convex Hull of Ellipsoids}
To illustrate the proposed approach, we consider a nonlinear system described by the following state-space representation
\begin{equation}
    x_{k+1} = A x_k + B u_k + f(x_k),
\end{equation}
where
\[
A = \begin{bmatrix}
0.7 & 0.1 & 0.05 \\ 
0 & 0.8 & 0.1 \\ 
0.1 & 0 & 0.6 
\end{bmatrix},
\quad
B = \begin{bmatrix}
0.5 & 0.1 \\ 
0.1 & 0.5 \\ 
0.2 & 0.2
\end{bmatrix}.
\]
and
\begin{equation}
    f(x_k) = \begin{bmatrix}
    \sin(x_{1,k}) \\ 
    x_{2,k}^2 - 0.5 \\ 
    \exp(-x_{3,k}) - 1
    \end{bmatrix}.
\end{equation}

The admissible set for all three states is defined as
\[
x_1, x_2, x_3 \in [-1, 1].
\]
This defines a bounded polyhedral region in the state space.

To align the ellipsoids with the geometry of the admissible region, we choose three principal directions and we construct three ellipsoids aligned with these directions. 

Using Algorithm~\ref{alg.set}, the convex hull of these ellipsoids is partitioned into regions that provide an efficient approximation of the admissible set.

Figure~\ref{fig:convex_hull} shows the convex hull formed by the three ellipsoids, while Figure~\ref{fig:polyhedral_approximation} shows the polyhedral under-approximation of the convex hull.

\begin{figure}[h]
    \centering
    \includegraphics[width=1\linewidth]{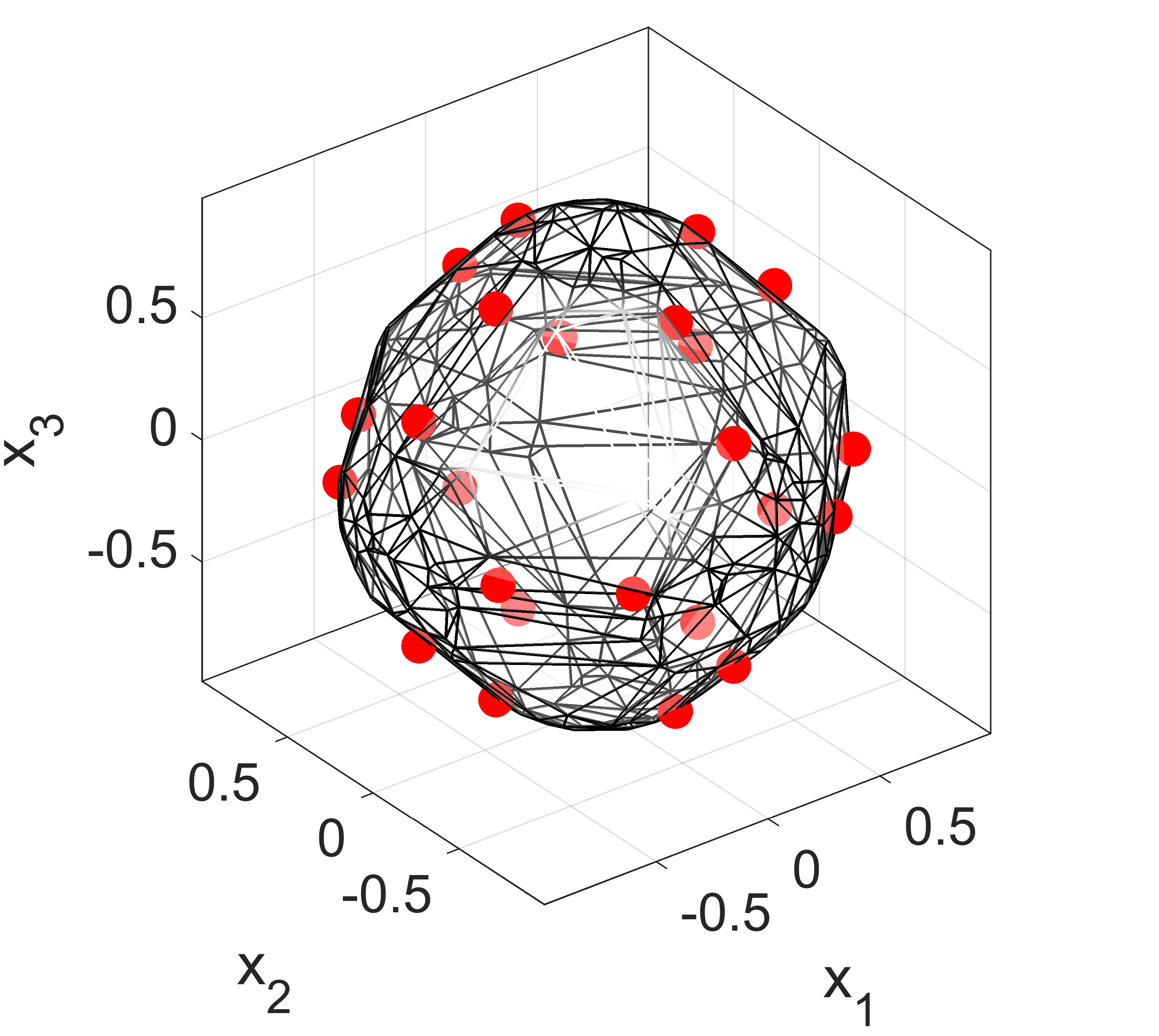}
    \caption{Convex hull of the three ellipsoids: The convex hull encloses the admissible region, with ellipsoids aligned to the principal directions.}
    \label{fig:convex_hull}
\end{figure}

\begin{figure}[h]
    \centering
    \includegraphics[width=1\linewidth]{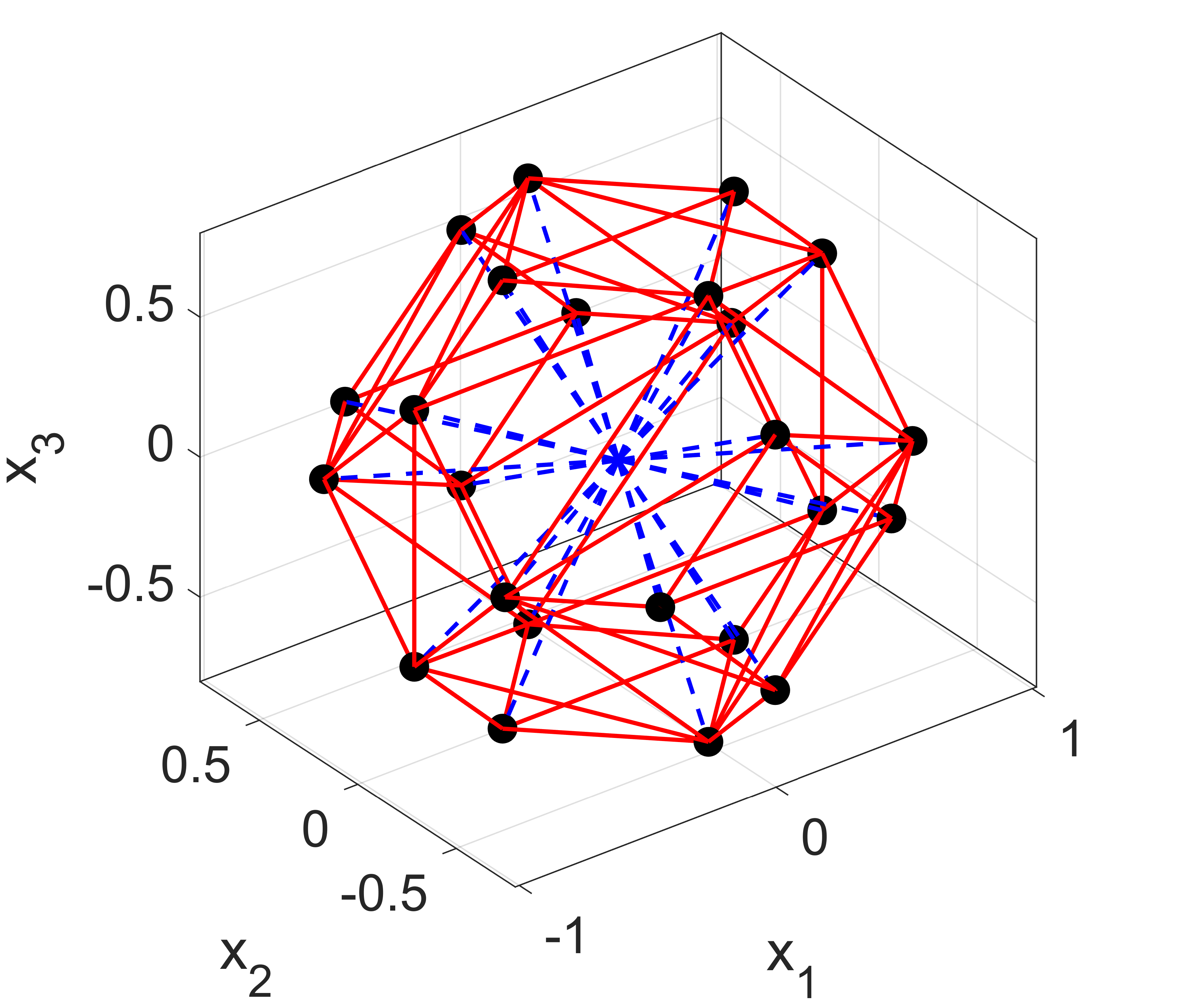}
    \caption{Polyhedral under-approximation of the convex hull.}
    \label{fig:polyhedral_approximation}
\end{figure}

The polyhedral under-approximation offers a conservative yet computationally efficient representation, which can be further leveraged for local nonlinear state-feedback control design and efficient intersection checking of convex hulls, as discussed later.

\subsubsection{Computation of State-Feedback Gains}
Once the convex hull is partitioned, the next step is to compute the state-feedback control gains for each partition. To compute the control gain for the partition, we first determine which ellipsoid each vertex belongs to. Then, we construct the gain interpolation using a simplex-based approach

\begin{equation}\label{eq.K_interpolation}
\mathcal{K}_{i} = \mathcal{K}_{\text{matrix}} V_{\text{matrix}}^{-1}, \quad \text{for} \;\;\; i = 1,\ldots,n_p
\end{equation}
where $n_p$ denotes the number of partitions,
\begin{equation}
V_{\text{matrix}} = [v_{e,1}^*, v_{e,2}^*, \dots, v_{e,n+m+2}^*],
\end{equation}
and
\begin{equation}
\mathcal{K}_{\text{matrix}} = \left[
\mathcal{K}_{\mathcal{E}_1} v_{e,1}^*, \quad
\mathcal{K}_{\mathcal{E}_2} v_{e,2}^*, \quad
\dots, \quad
\mathcal{K}_{\mathcal{E}_{n+m+2}} v_{e,n+m+2}^*
\right]
\end{equation}

where $\mathcal{K}_{\mathcal{E}_i}$ represents the control gain associated with the ellipsoid $\mathcal{E}_i$ to which vertex $v_{e,i}^*$ belongs. The final control law is given by
\begin{equation}\label{eq.u_final}
v_{i,k} = \mathcal{K}_{i} \mathcal{Z}_k, \quad \text{for} \;\;\; i=1,\dots,n_p.
\end{equation}

The process is summarized in Algorithm \ref{Algorithm_2}, which is executed online.

\begin{algorithm}
\caption{State-Feedback Gain Computation}\label{Algorithm_2}
\begin{algorithmic}[1]
\Require Current state $\zeta_{e,k}$; Number of partitions $n_p$; Control gains $\mathcal{K}_{\mathcal{E}_i}$ for $i=1,2,\ldots,n_e$.
\Ensure State-feedback gain $\mathcal{K}_{j}$ for the $j$th partition.

\State \textbf{Determine} the partition to which $\zeta_k$ belongs.
\State \textbf{Identify} the corresponding convex hull vertices $\{ v_{e,1}^*, v_{e,2}^*, \dots, v_{e,n+m+2}^* \}$.
\State \textbf{Compute} the interpolated control gain
\begin{equation}
\mathcal{K}_{j} = K_{\text{matrix}} V_{\text{matrix}}^{-1}.
\end{equation}

\State \textbf{Obtain} the control law
\begin{equation}
v_{j,k} = \mathcal{K}_{j} \mathcal{Z}_k.
\end{equation}
 
\end{algorithmic}
\end{algorithm}

\subsection{Safe Transition Between Convex Hull of Ellipsoids}



A key component of data-driven motion planning using convex hulls of ellipsoids is ensuring safe transitions between successive convex regions while preserving system invariance. Unlike the single ellipsoid case, transitions here require checking for intersections between convex hulls composed of multiple ellipsoids and selecting an intermediate point to enable smooth progression.

Since safety constraints are primarily evaluated in the \( x \)-\( y \) plane, the ellipsoids forming each convex hull are first projected onto this plane. The 2D convex hulls are then computed, and their intersection is checked by verifying whether any vertex of one hull lies within the other. If an intersection exists, a vertex common to both regions is selected as the intermediate point, ensuring safe and continuous navigation.

Algorithm~\ref{alg:check_intersection_2d} formalizes this process, including projection, convex hull construction, intersection checking, and intermediate point selection. This method provides a dynamically feasible and safe transition strategy while enabling efficient exploration of the admissible state space.

\begin{algorithm}
\caption{Safe Transition Between Projected Convex Hull of Ellipsoids}
\label{alg:check_intersection_2d}
\textbf{Input:} Convex hulls of projected ellipsoids $\mathcal{F}_{proj,1}$ and $\mathcal{F}_{proj,2}$, corresponding vertices $v_{e,1}^{proj}$ and $v_{e,2}^{proj}$.\\
\textbf{Output:} Selected intermediate vertex $v_{\text{selected}}$, overlap flag $\text{overlap}$.

\begin{algorithmic}[1]
\State \textbf{Initialize} $\text{overlap} \gets 0$, $v_{\text{selected}} \gets \emptyset$
\State \textbf{Project} the shape matrices of all ellipsoids forming $\mathcal{F}_1$ and $\mathcal{F}_2$ onto the $x$-$y$ plane to obtain projected ellipsoids.
\State \textbf{Generate} the 2D convex hulls $\mathcal{F}_{proj,1}$ and $\mathcal{F}_{proj,2}$ using the projected ellipsoid vertices.
\State \textbf{Check} if any vertex of $\mathcal{F}_{proj,2}$ is inside $\mathcal{F}_{proj,1}$:
\If{$\mathcal{F}_{proj,1} v_{e,2}^{proj} \leq \bar{g}_{proj,1}$ for any vertex $v_{e,2}^{proj}$}
    \State \textbf{Select} the first valid vertex from $v_{e,2}^{proj}$ as $v_{\text{selected}}$
    \State $\text{overlap} \gets 1$
    \State \textbf{Return} $(v_{\text{selected}}, \text{overlap})$
\EndIf
\State \textbf{Check} if any vertex of $\mathcal{F}_{proj,1}$ is inside $\mathcal{F}_{proj,2}$:
\If{$\mathcal{F}_{proj,2} v_{e,1}^{proj} \leq \bar{g}_{proj,2}$ for any vertex $v_{e,1}^{proj}$}
    \State \textbf{Select} the first valid vertex from $v_{proj,1}$ as $v_{\text{selected}}$
    \State $\text{overlap} \gets 1$
\EndIf
\State \textbf{Return} $(v_{\text{selected}}, \text{overlap})$
\end{algorithmic}
\end{algorithm}

The system first moves from the parent node toward the selected vertex (intermediate waypoint), using the control policy associated with the corresponding partition of the parent region. Once the system reaches the intermediate point, the control gain switches to the one associated with the appropriate partition in the child region, guiding the system to the child node. This approach generalizes the safe transition strategy beyond single ellipsoids and enables more flexible and less conservative navigation, particularly in complex environments where a single ellipsoid may not adequately represent the admissible region.


\subsection{Execution of Safe Motion Planning}

Once a safe path has been planned using the convex hull of ellipsoids, the next step is to execute the planned trajectory while ensuring that the system remains within the predefined safe regions at each step. This involves continuously checking the current sampled point to determine which partition (convex region) it belongs to and applying the corresponding control gain.

The control law for each partition is computed based on the associated gain. The system state is updated using the dynamics and the computed control input. The process is repeated until the system reaches the target node $p_{goal}$.

Algorithm~\ref{alg:execution_convex_hull} describes the execution process. The algorithm identifies the current partition at each step and retrieves the corresponding control gain to compute the control input. This ensures that the trajectory remains safe and dynamically feasible throughout the execution.



\begin{algorithm}
\caption{Execution of Safe Motion Planning with Convex Hull of Ellipsoids}
\label{alg:execution_convex_hull}
\textbf{Input:} Precomputed waypoint sequence \( \{p_1, p_2, \dots, p_{N_w}\} \), convex partitions \( \{\mathcal{C}_i\} \), control gains \( \{\mathcal{K}_i\} \), system dynamics \( f(x, u) \), initial position \( p_{init} \), goal position \( p_{goal} = p_{N_w} \).\\
\textbf{Output:} Safe trajectory execution from \( p_{init} \) to \( p_{goal} \).

\begin{algorithmic}[1]
\State \textbf{Initialize} current waypoint index \( i \gets 1 \), current position \( p_{curr} \gets p_{init} \)
\While{$p_{curr} \neq p_{N_w}$}
    \State \textbf{Identify} the current convex partition \( \mathcal{C}_i \) containing \( p_{curr} \)
    \State \textbf{Retrieve} the control gain \( \mathcal{K}_i \) associated with \( \mathcal{C}_i \)
    \State \textbf{Compute} virtual control input \( v_k = \mathcal{K}_i \mathcal{Z}_k \)
    \State \textbf{Update} actual control input via integration: \( u_k = u_{k-1} + v_k \)
    \State \textbf{Update} system state: \( x_{k+1} = f(x_k, u_k) \)
    \State \textbf{Update} position \( p_{curr} \gets \text{position component of } x_{k+1} \)
    \If{$\|p_{curr} - p_i\| < r_f$ \textbf{and} \( i < N_w \)}
        \State \textbf{Advance} to the next waypoint: \( i \gets i + 1 \)
    \EndIf
\EndWhile
\State \textbf{Return} executed trajectory
\end{algorithmic}
\end{algorithm}

\subsection{Overall Method for Data-Driven Motion Planning Using the Convex Hull of Ellipsoids}



This subsection summarizes the complete framework for data-driven motion planning using convex hulls of ellipsoids. The method integrates key components, including admissible set identification, convex hull construction, safe transition verification, and trajectory execution. Leveraging convex hull representations enhances flexibility and ensures dynamically feasible paths.

A major advantage of using convex hulls of ellipsoids is the ability to expand the set of feasible transitions, enabling smoother and less conservative trajectories compared to single-ellipsoid approaches. This is accomplished by detecting intersections between adjacent convex hulls and computing corresponding control gains to guide the system safely.

Algorithm~\ref{alg:CHE-RRT} outlines the overall procedure. The algorithm initializes a graph and constructs the initial convex hull of ellipsoids. It then iteratively samples new points, determines admissible sets, computes convex hulls, and verifies safe transitions via intersection checks. Control gains are assigned to each segment to ensure safe and feasible execution of the planned path.

\begin{algorithm}
\caption{Data-Driven Safe Motion Planning Using Convex Hulls of Ellipsoids}
\label{alg:CHE-RRT}
\begin{algorithmic}[1]
\Require $p_{init}$, $p_{goal}$, collected data $V_0$, $\Xi_1$, and $\mathcal{Z}_0$, contraction factor $\lambda$, maximum iterations $N_{\text{max}}$, probability $p_r$ of goal sampling.
\Ensure A safe, dynamically feasible path from $p_{init}$ to $p_{goal}$.

\State \textbf{Initialize} graph $G = (V, E)$ with $V = \{p_{init}\}$, $E = \emptyset$.
\State \textbf{Compute} the initial convex hull of ellipsoids for $p_{init}$ using the SDP \eqref{eq.SDP_CHE}.
\While{$\|p_s - p_{goal}\| \geq r_f$ \textbf{and} iteration $\leq N_{\text{max}}$}
    \State \textbf{Sample} a new point $p_{s}$ using probability $p_r$.
    \State \textbf{Identify} the admissible set around $p_{s}$.
    \State \textbf{Compute} the convex hull of ellipsoids for $p_{s}$ by solving the SDP \eqref{eq.SDP_CHE}.
    \State \textbf{Check} for safe transition by verifying the intersection between the convex hulls at $p_{s}$ and $p_{nearest}$ using Algorithm 7.
    
    \If{an intersection exists}
        \State \textbf{Add} $p_{s}$ and the intermediate node to $G$.
    \EndIf
\EndWhile
\State \textbf{Execute} the planned trajectory using system dynamics and associated control gains for each convex hull.
\State \Return executed trajectory
\end{algorithmic}
\end{algorithm}

\section{Simulation}
In this section, the proposed motion planning algorithm is implemented on a simulated model of a real-world autonomous ground vehicle—the ROSbot 2R—within the Gazebo simulation environment, as shown in Fig.~\ref{fig:rosbot_env}. The Robot Operating System (ROS) is used to coordinate communication, control, and data exchange between the planner and the robot. Figure~\ref{fig:rosbot_env}(a) shows the initial configuration of the robot, while Fig.~\ref{fig:rosbot_env}(b) depicts the simulation environment, where a red cube shows the location of a static obstacle that the robot must safely avoid during navigation.

The kinematic model of the ROSbot 2R is described by the following equations
\begin{equation}\label{eq.kinematic_model} 
\begin{aligned}
\dot{x} &= v \cos \theta, \\ 
\dot{y} &= v \sin \theta, \\ 
\dot{\theta} &= \omega_{b},
\end{aligned} 
\end{equation}

In this model, the variables \( x \) and \( y \) represent the robot's position, and \( \theta \) denotes its orientation. The inputs \( v \) and \( \omega_b \) correspond to the linear and angular velocities of the robot, respectively. Let the state and input vectors be defined as \( x_s = [x, y, \theta]^T \) and \( u = [v, \omega_b]^T \). The discrete-time kinematic model can then be expressed as \eqref{eq.system}.

In the simulation, the contraction rate is set to \( \lambda = 0.84 \), and other parameters are selected as \( \epsilon_i = 0.002 \) and \( \tau_i = 0.1 \) for $i=1,\ldots,7$. For the lifted representation we set $\xi=[x_s^{\top}\;u^{\top}]^{\top}$ and adopt the dictionary
$S(\xi)=\bigl[v\cos\theta,\;v\sin\theta\bigr]^{\top}$.  
This yields a lifted dimension of $n_\mathcal{z}=7$, so by Assumption 5 at least $N\ge n_z+1=8$ state–input samples are required to satisfy the rank condition. Each sample is perturbed with zero‑mean Gaussian noise, $\Sigma_{\text{noise}}=10^{-4}I$, ensuring that the data matrix pair $\mathcal{Z_0}$ remains full row rank while still reflecting realistic measurement uncertainty. The robot starts at \( (-40,-40) \), must reach the goal \( (40,40) \), and must avoid a red square obstacle of side length \( 16\;\text{m} \) that is centered at the origin \( (0,0) \).


Figures \ref{fig:path_che}–\ref{fig:sample_bar} compare the three planners we tested in MATLAB. Figure \ref{fig:path_che} shows the path found by our convex-hull-of-ellipsoids (CHE) method; the polygons mark safe regions that the robot can stay inside while the feedback gains change from point to point. Figure \ref{fig:path_ref16} is the result of the “containment’’ planner from \cite{niknejad2024soda}, which insists that every new safe set must fully contain the center of the previous one. Figure \ref{fig:path_overlap} uses our overlapping-ellipsoids idea: it lets neighboring safe sets merely intersect and still proves safety with a small extra check, so the robot needs far fewer sample points to complete the same maneuver. The bar chart in Fig. \ref{fig:sample_bar} counts those samples: the containment method uses the most, the overlapping-ellipsoid version uses less, and the CHE planner needs the least, showing that relaxing the containment rule—or switching to CHE—cuts sampling effort without giving up safety. Among the three, the CHE planner was selected for deployment in the real-time robot simulation.

\begin{figure}[h]
    \centering
    \includegraphics[width=0.5\textwidth]{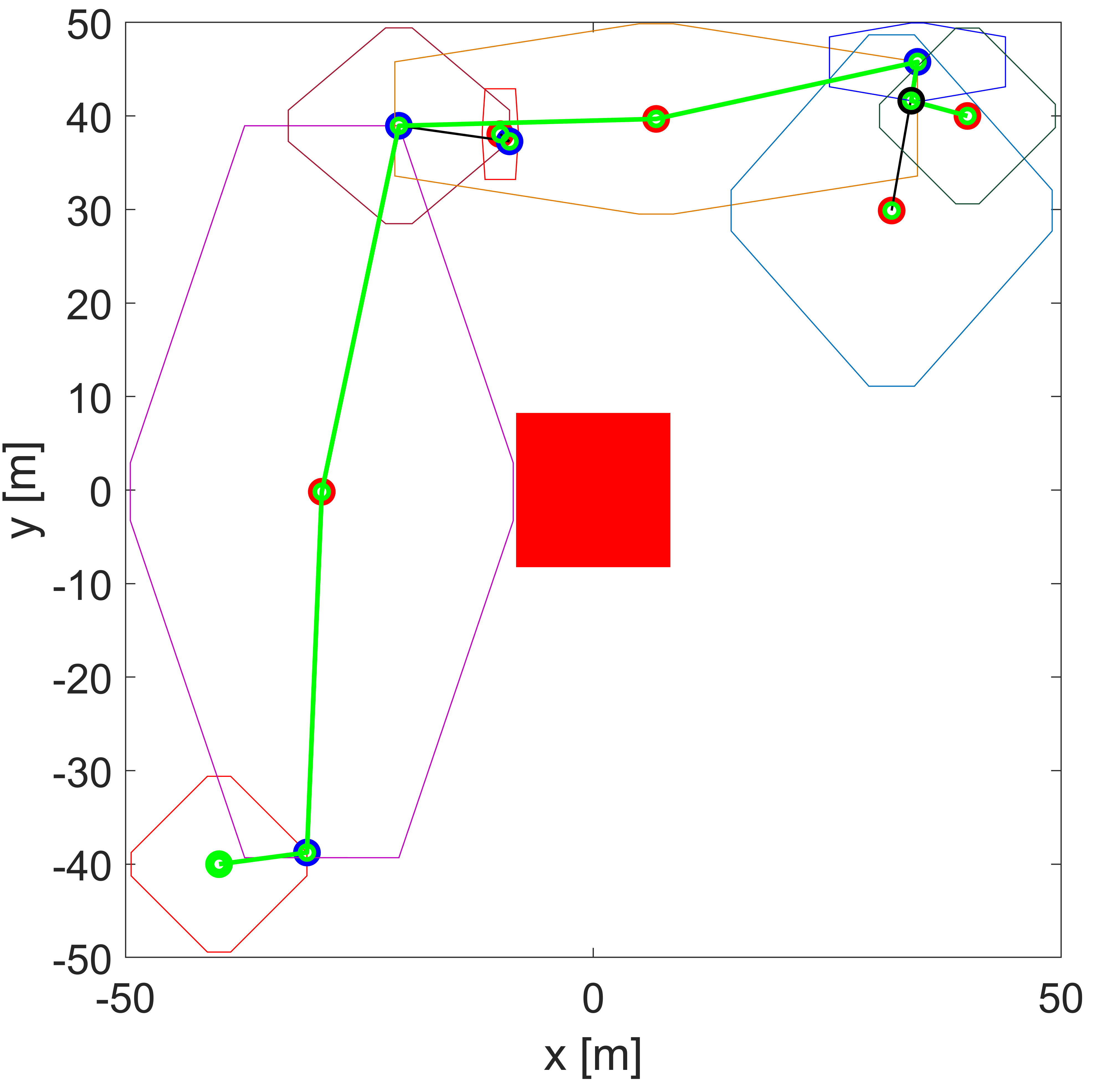}
    \caption{Safe path produced by the convex-hull of ellipsoids (CHE) method.}
    \label{fig:path_che}
\end{figure}

\begin{figure}[h]
    \centering
    \includegraphics[width=0.5\textwidth]{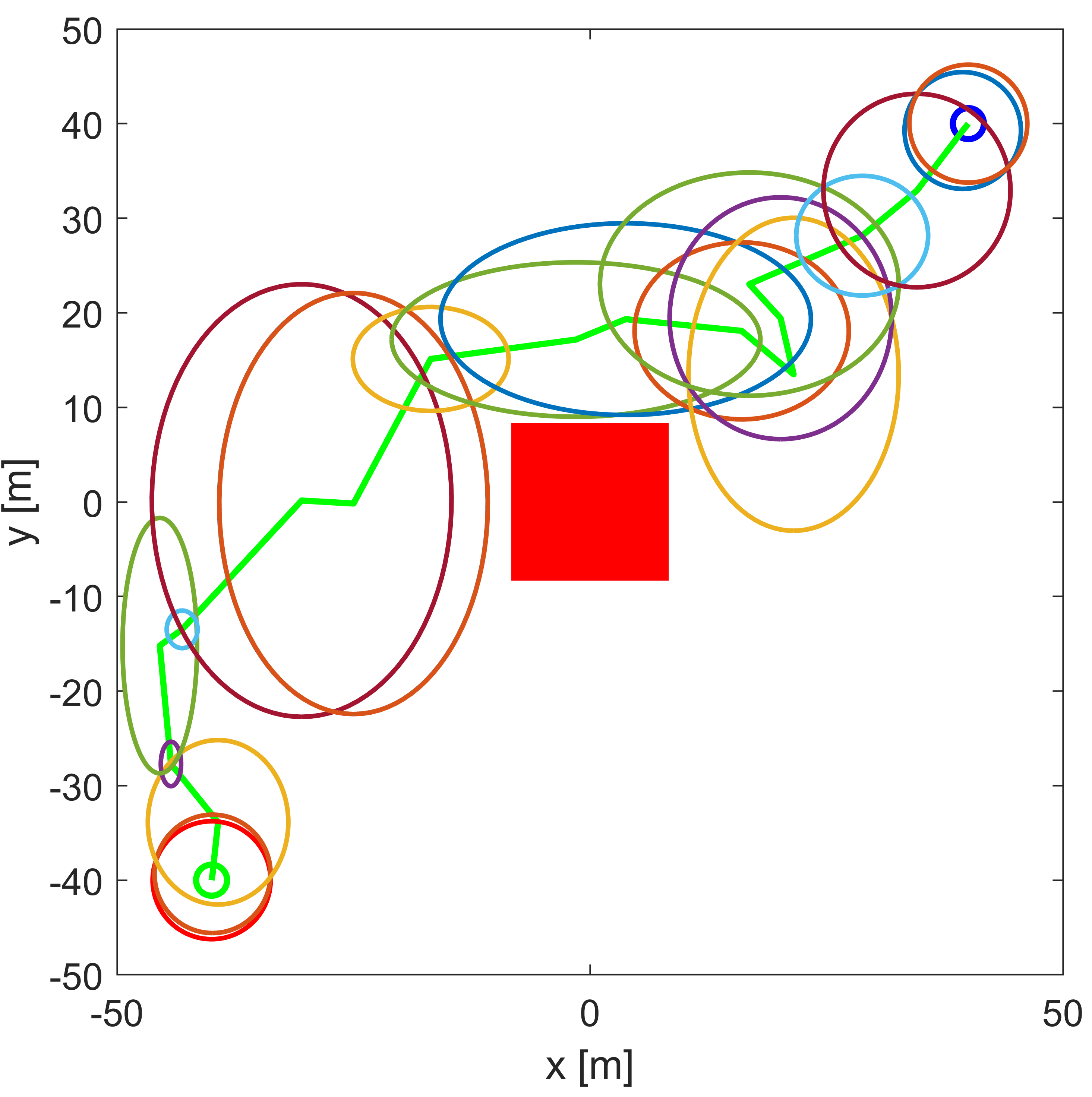}
    \caption{Safe path produced by the proposed overlapping-ellipsoids method.}
    \label{fig:path_overlap}
\end{figure}

\begin{figure}[h]
    \centering
    \includegraphics[width=0.5\textwidth]{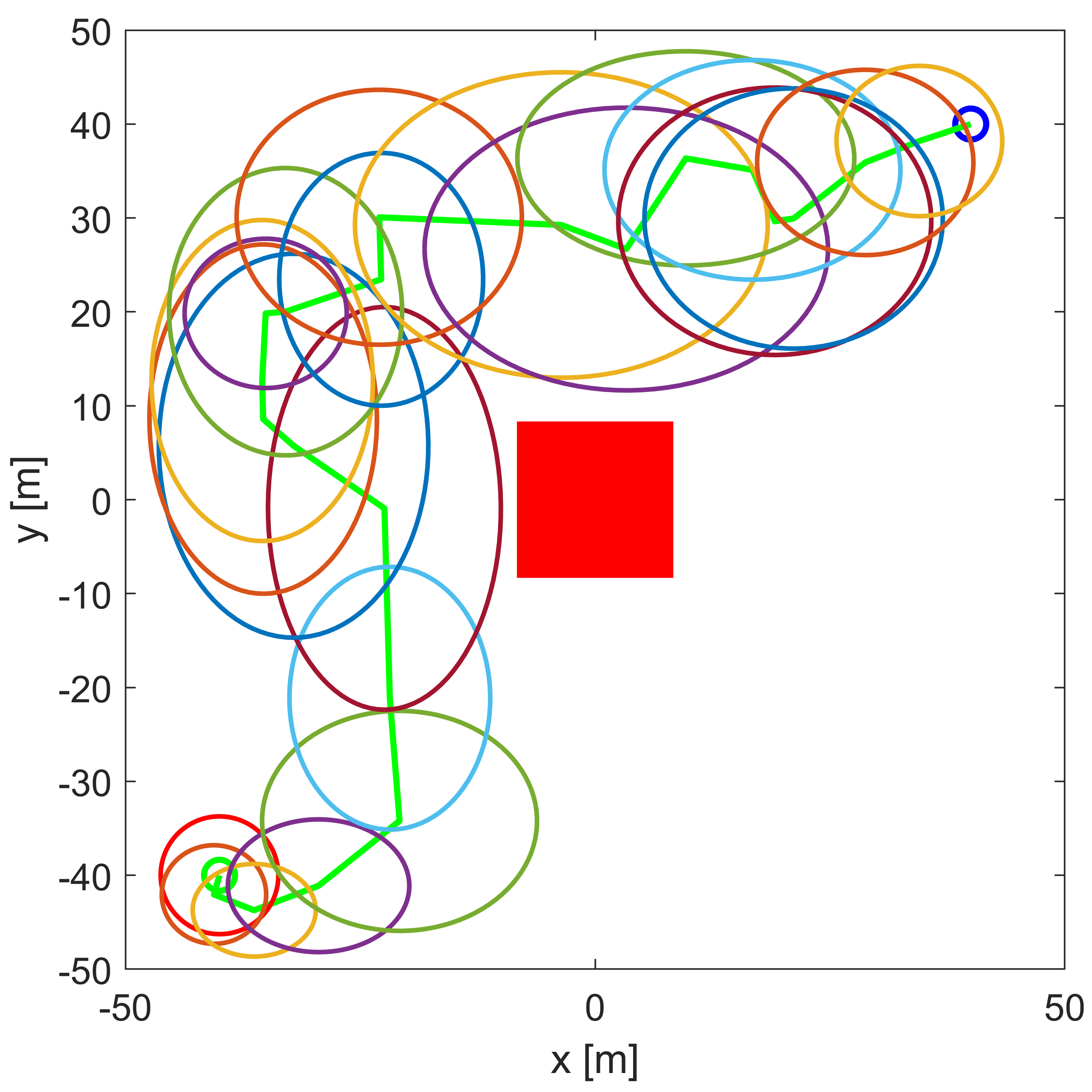}
    \caption{Safe path produced by the containment-based planner of~\cite{niknejad2024soda}.}
    \label{fig:path_ref16}
\end{figure}

\begin{figure}[h]
    \centering
    \includegraphics[width=0.5\textwidth]{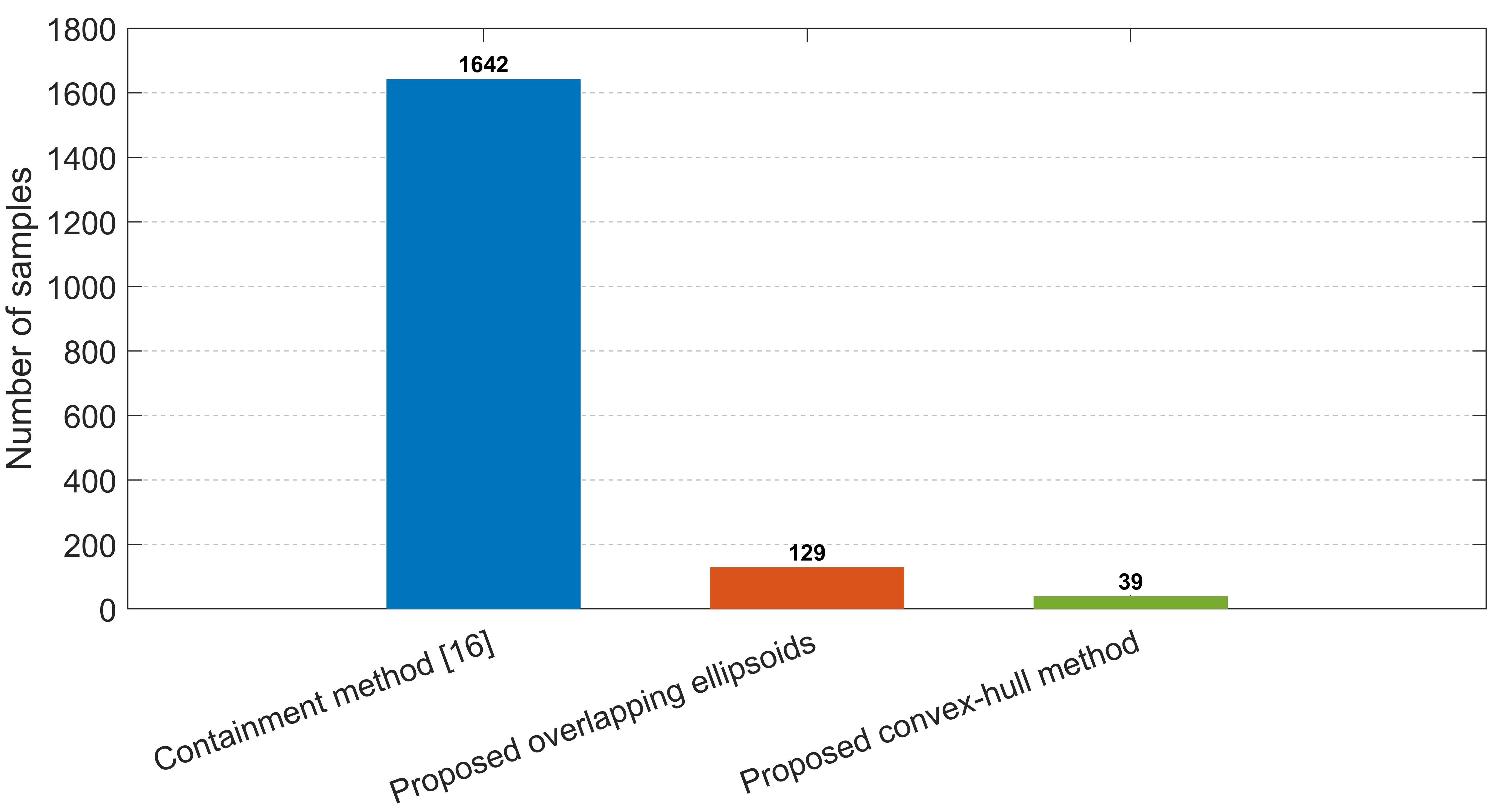}
    \caption{Number of samples used by each planner (containment~\cite{niknejad2024soda}, overlapping ellipsoids, and CHE).}
    \label{fig:sample_bar}
\end{figure}

Figure 8 shows the Gazebo test world, with the obstacle layout and the ROSbot 2R posed at the start. During the experiment a Python node running under ROS 2 streams the pre-computed CHE way-points and feedback gains to Gazebo, reads the robot’s odometry, and updates the control inputs at each cycle, closing the loop in real time. Figure 9 collects snapshots of the resulting motion: the robot stays inside every certified safe region defined by the convex-hull planner, navigates around the obstacle without incident, and reaches the goal while respecting all state constraints. This simulation confirms that the proposed method, paired with a straightforward ROS-Python interface, transfers seamlessly from offline planning to physics-based execution and maintains both safety and stability throughout.

Throughout the simulation, the control strategy ensures smooth and collision-free navigation. Since all optimization and data collection steps are performed offline, the online implementation in Gazebo remains computationally efficient, making the method suitable for real-time applications.


\begin{figure}[h]
    \centering
    \begin{subfigure}[b]{0.9\linewidth}
        \centering
        \includegraphics[width=0.7\linewidth]{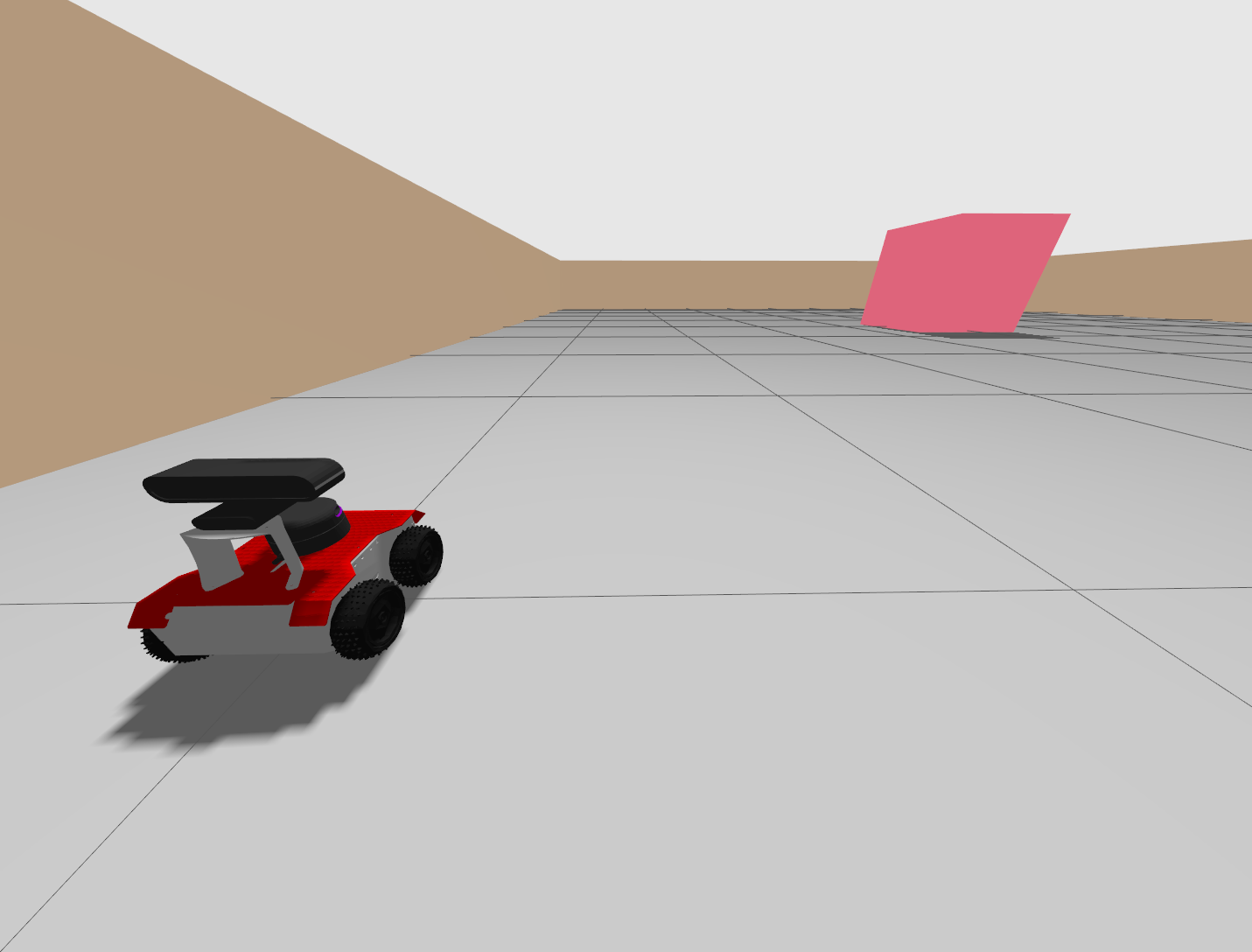}
        \caption{}
        \label{fig:rosbot_env_1}
    \end{subfigure}
    
    \vspace{0.5cm} 
    
    \begin{subfigure}[b]{0.9\linewidth}
        \centering
        \includegraphics[width=0.7\linewidth]{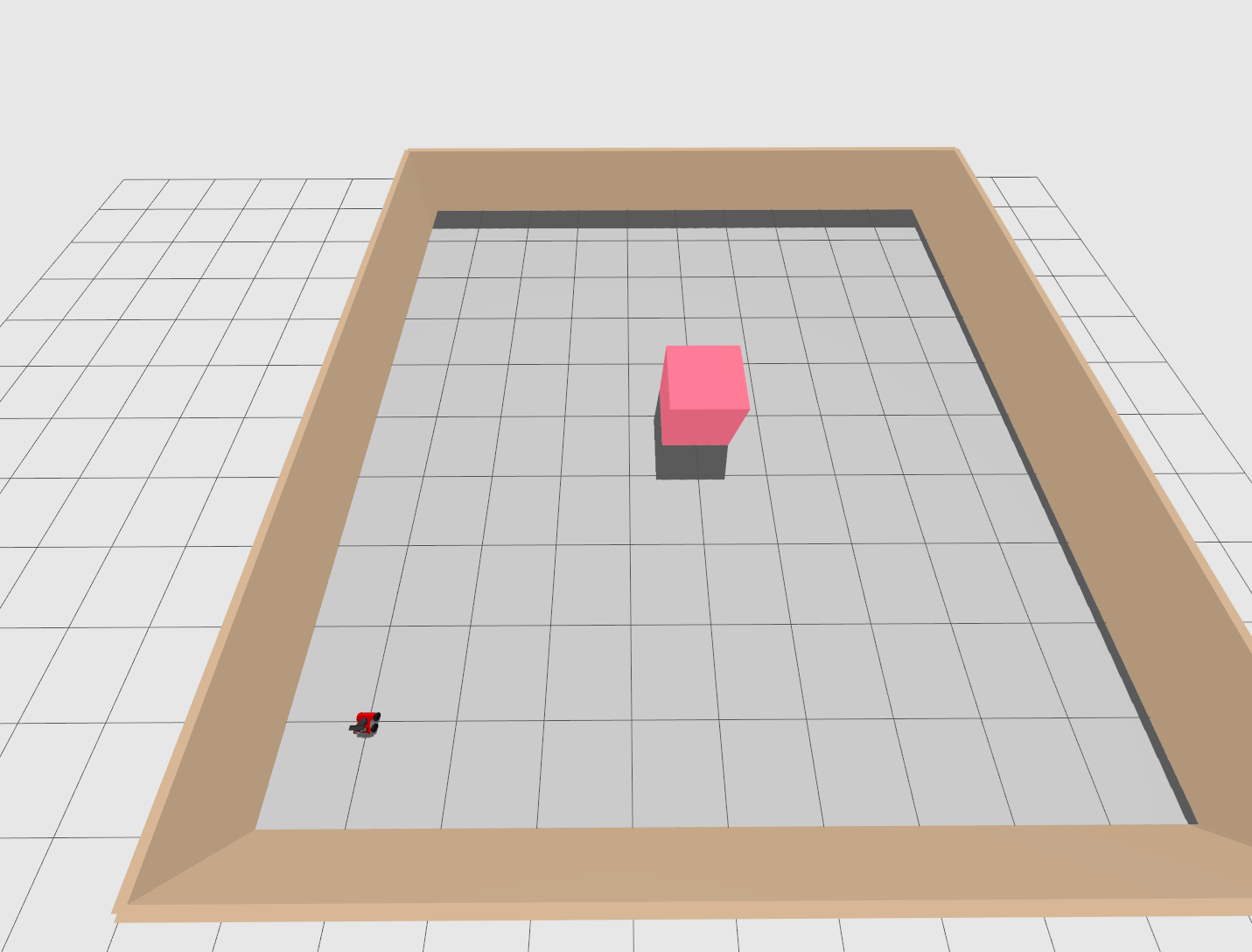}
        \caption{}
        \label{fig:rosbot_env_2}
    \end{subfigure}

    \caption{The robot's simulated environment in Gazebo, where the red cube represents a static obstacle used to evaluate collision avoidance during motion execution.}
    \label{fig:rosbot_env}
\end{figure}

\begin{figure*}[!h] 
    \centering
    \begin{tabular}{c c}
        \begin{subfigure}[b]{0.45\linewidth}
            \centering
            \includegraphics[height=4.0cm,width=0.7\linewidth]{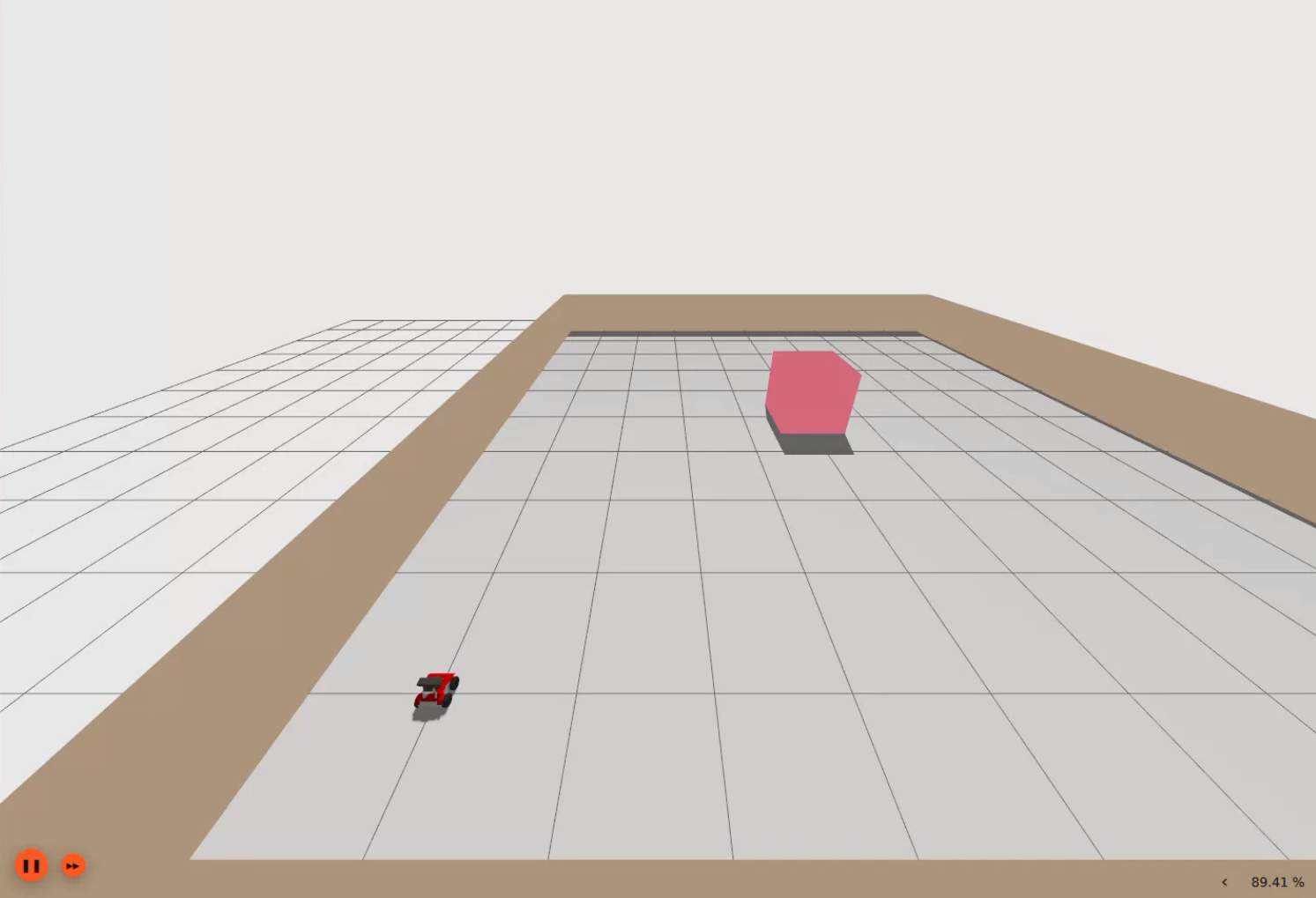}
            \caption{Episode 1}
            \label{fig:gazebo_1}
        \end{subfigure} &
        \begin{subfigure}[b]{0.45\linewidth}
            \centering
            \includegraphics[height=4.0cm,width=0.7\linewidth]{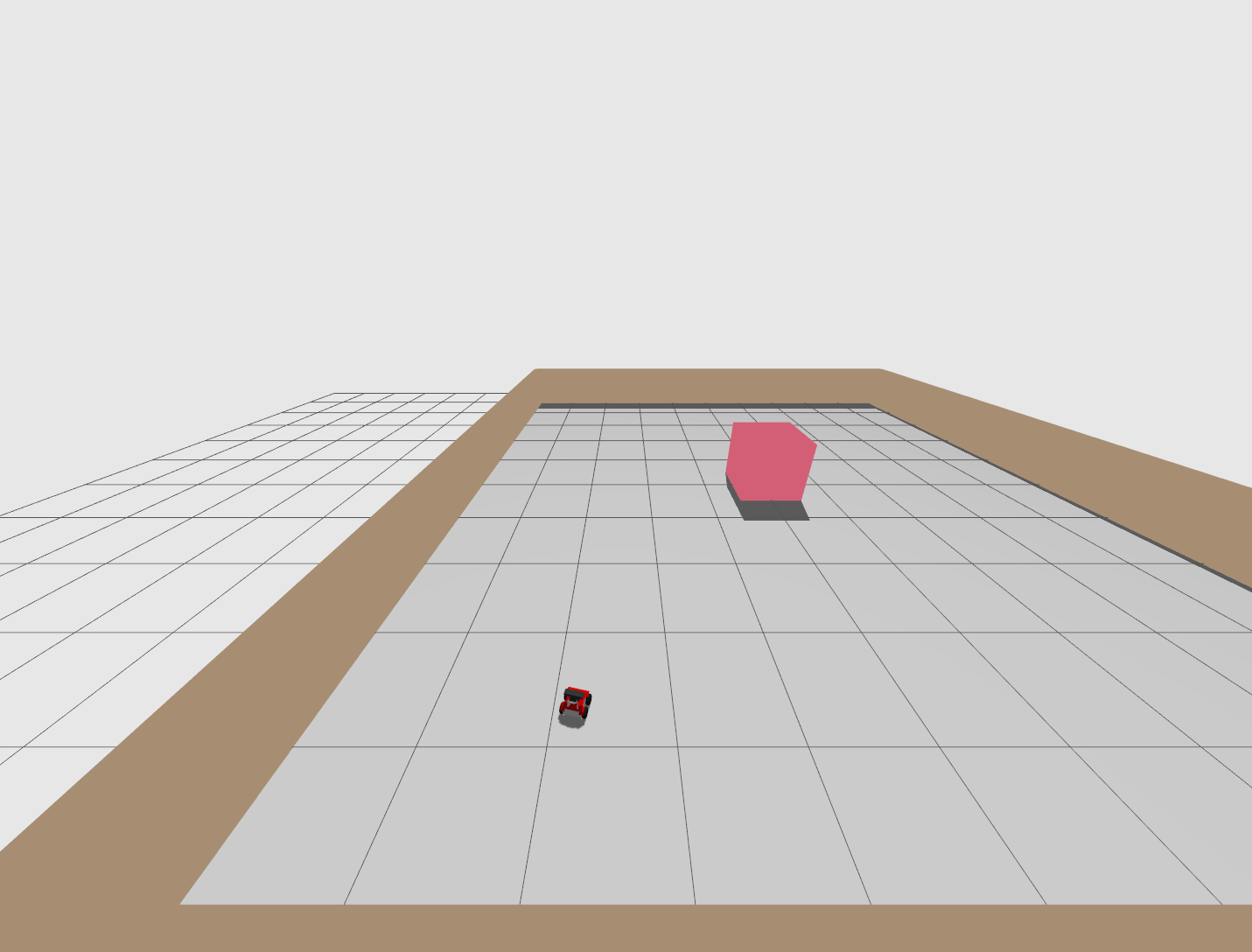}
            \caption{Episode 2}
            \label{fig:gazebo_2}
        \end{subfigure} \\

        \begin{subfigure}[b]{0.45\linewidth}
            \centering
            \includegraphics[height=4.0cm,width=0.7\linewidth]{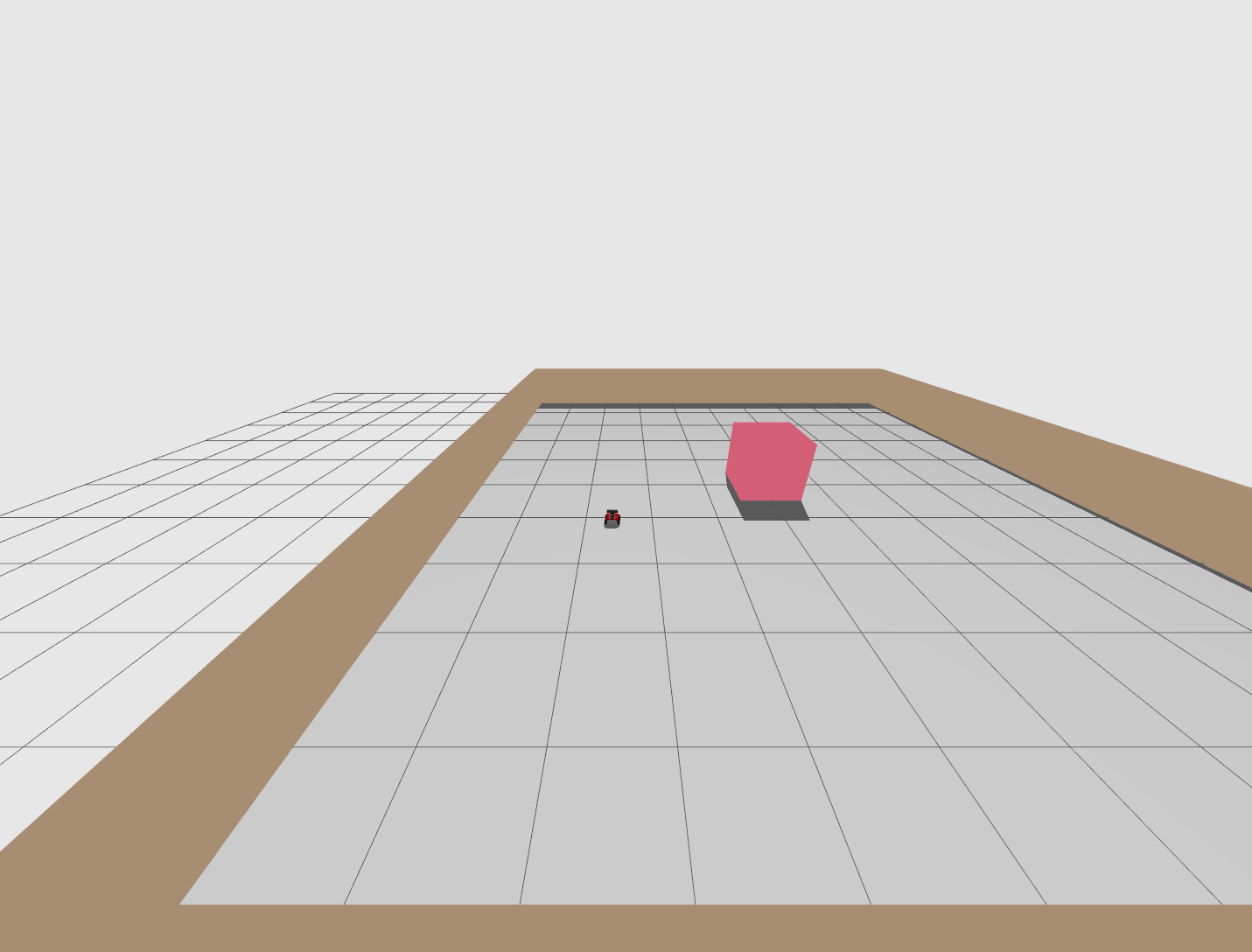}
            \caption{Episode 3}
            \label{fig:gazebo_3}
        \end{subfigure} &
        \begin{subfigure}[b]{0.45\linewidth}
            \centering
            \includegraphics[height=4.0cm,width=0.7\linewidth]{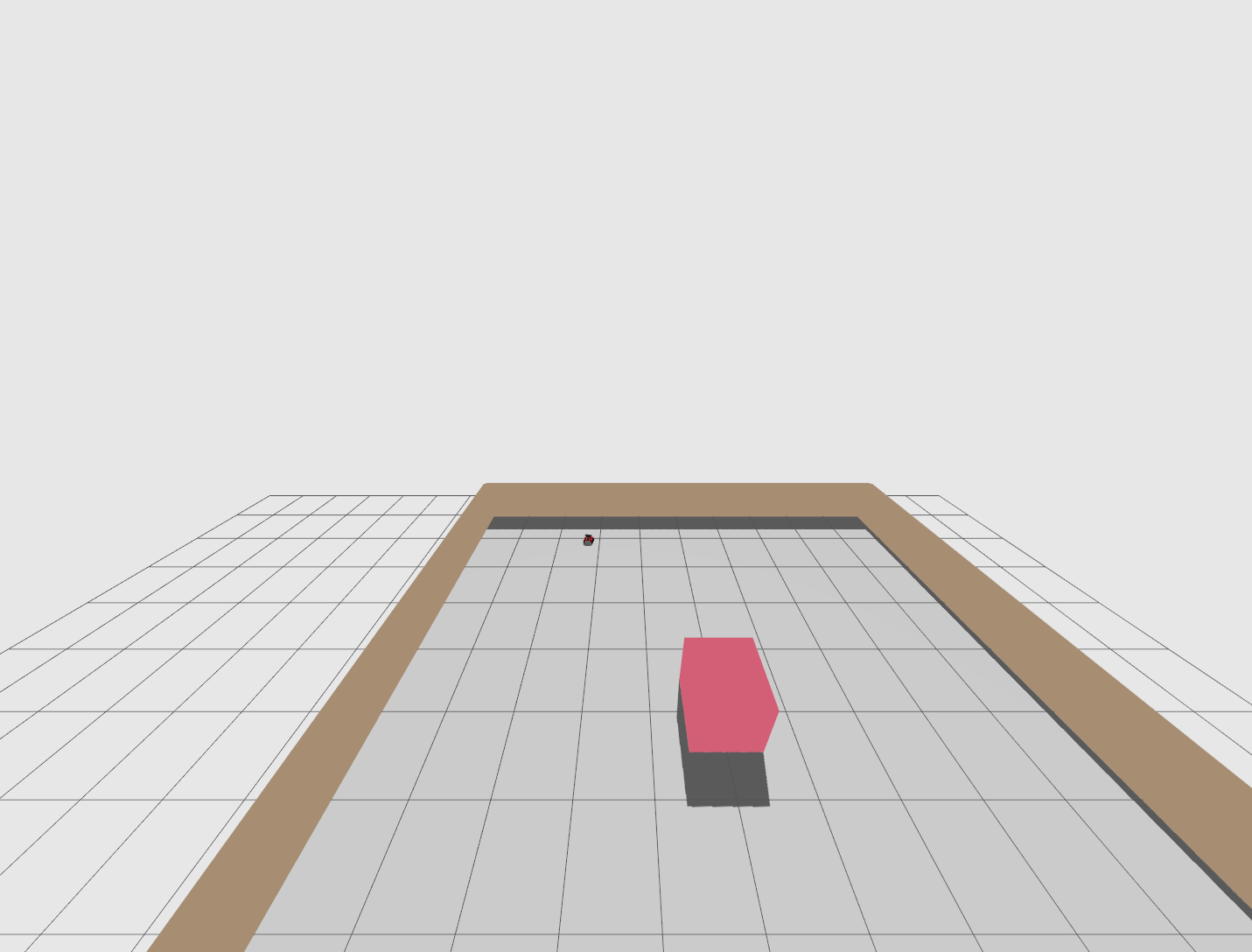}
            \caption{Episode 4}
            \label{fig:gazebo_4}
        \end{subfigure} \\

        \begin{subfigure}[b]{0.45\linewidth}
            \centering
            \includegraphics[height=4.0cm,width=0.7\linewidth]{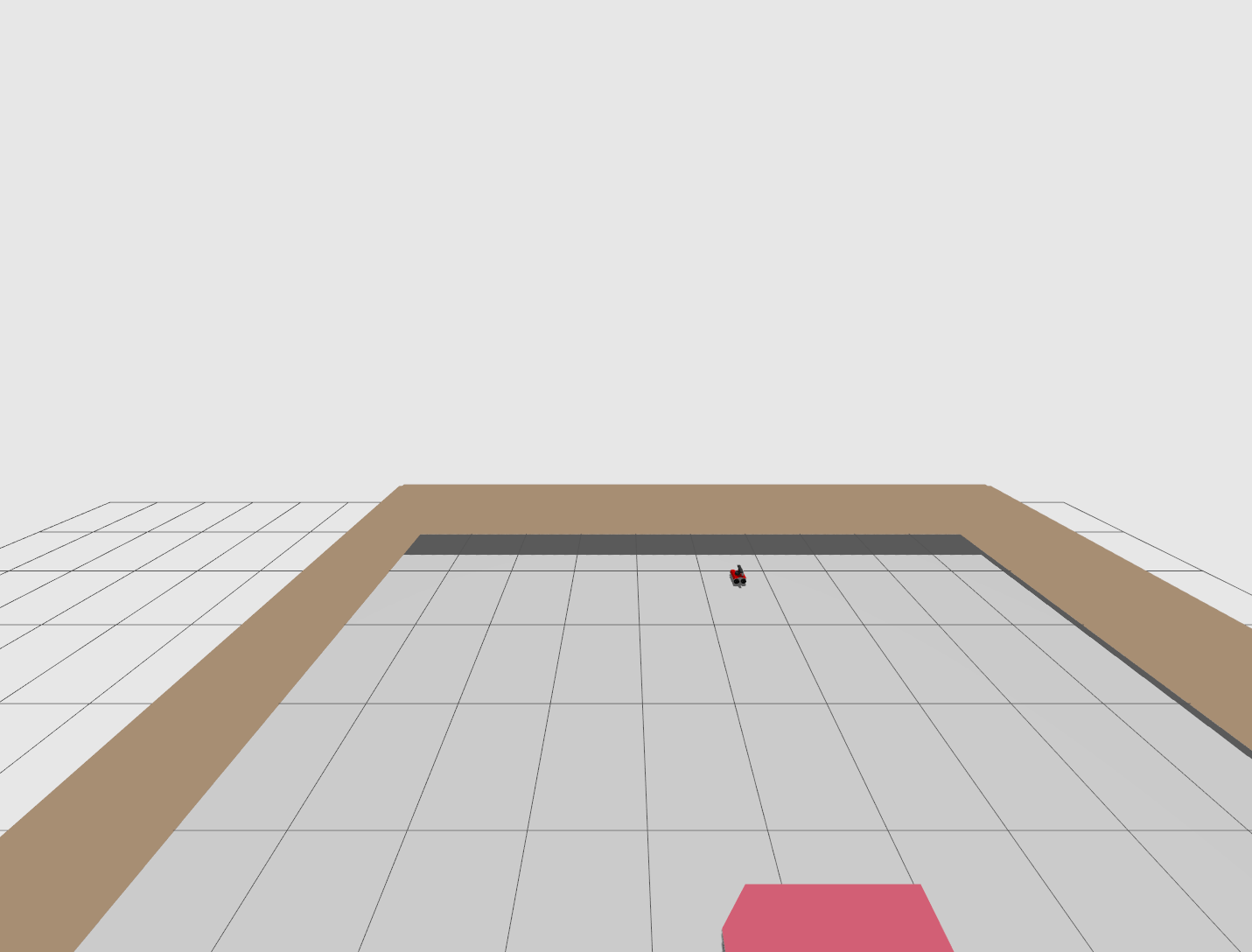}
            \caption{Episode 5}
            \label{fig:gazebo_5}
        \end{subfigure} &
        \begin{subfigure}[b]{0.45\linewidth}
            \centering
            \includegraphics[height=4.0cm,width=0.7\linewidth]{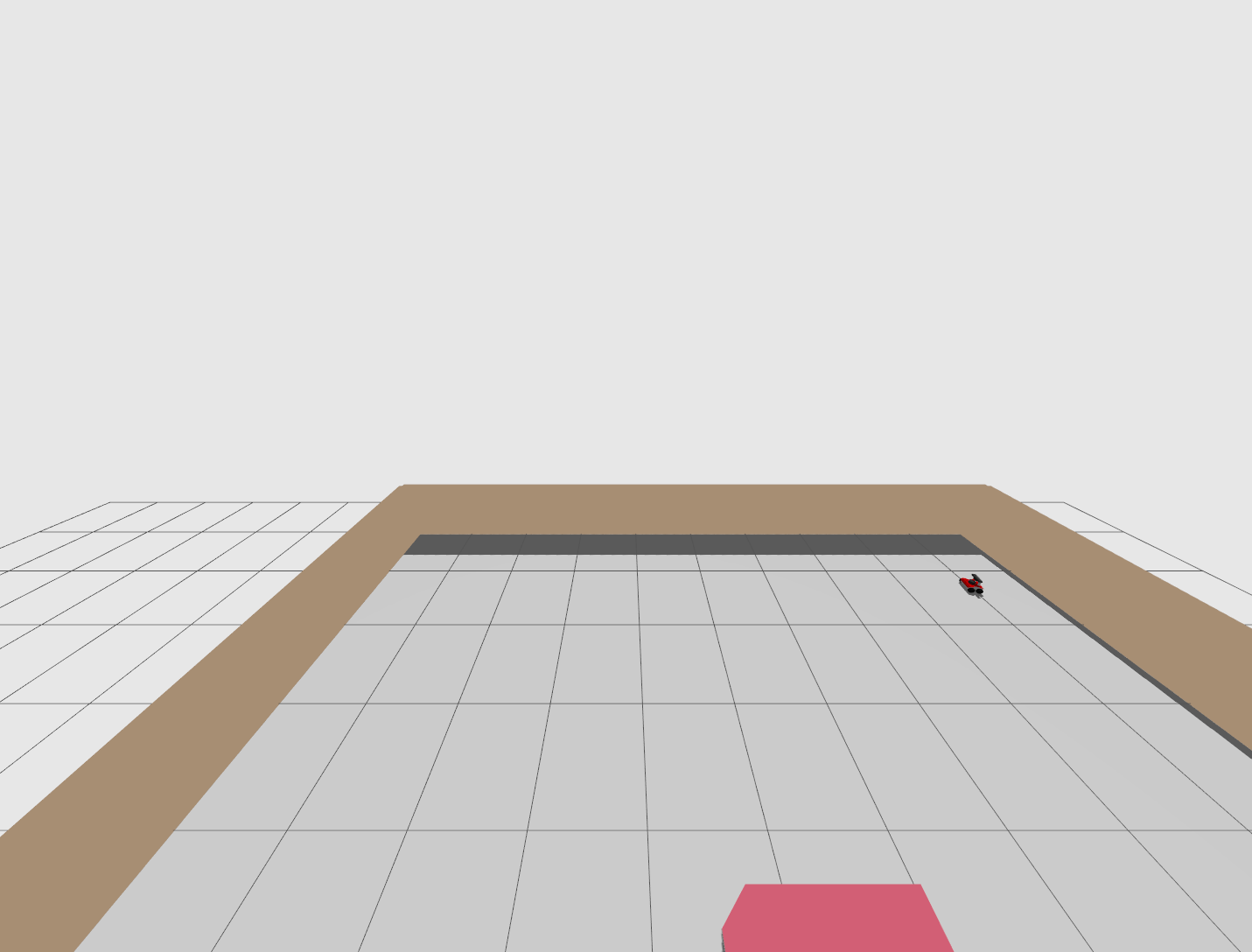}
            \caption{Episode 6}
            \label{fig:gazebo_6}
        \end{subfigure}
    \end{tabular}
    \caption{Different episodes of the Gazebo simulation. Each subplot shows a snapshot of the robot during the execution of the preplanned trajectory.}
    \label{fig:3by2grid}
\end{figure*}

\section{Conclusion}
In this work, a data-driven motion planning algorithm for nonlinear systems was developed to ensure safety by constructing invariant convex hulls of ellipsoids. The method leveraged data to define admissible polyhedral sets and determine safe regions without requiring an explicit model of the system dynamics. A key feature of the approach was verifying intersections between successive convex hulls to identify safe transitions and compute intermediate points for smooth, collision-free navigation. Control gains were interpolated to guide the system state within these safe regions while guaranteeing safety and avoiding obstacles. The proposed method effectively addressed the challenges of navigating complex non-convex environments and provided strong safety guarantees. Future work will focus on incorporating noise and uncertainties into the framework to provide probabilistic safety guarantees for more realistic and robust performance.

\appendices
\section{Proof of Theorem~\ref{thm:single_ellipsoid}}
\label{app:proof_single_ellipsoid}
To guarantee \( \lambda \)-contractivity, it is sufficient that
\begin{equation}
\zeta_{e,k+1}^\top \mathcal{P}^{-1} \zeta_{e,k+1} \leq \lambda \zeta^\top_{e,k} \mathcal{P}^{-1} \zeta_{e,k}.
\end{equation}

Upon expansion and simplification, the expression becomes
\begin{align}\label{eq.proof_CE_5}
& \lambda \zeta_{e,k}^\top \mathcal{P}^{-1} \zeta_{e,k} - \left(\Xi_1\mathcal{G}_{K,l}\zeta_{e,k}\right)^\top \mathcal{P}^{-1} \left(\Xi_1\mathcal{G}_{K,l}\zeta_{e,k}\right) \nonumber \\
& - \Gamma_1 - \Gamma_2 \geq 0,
\end{align}
with \( \Gamma_1 \) and \( \Gamma_2 \) defined by
\begin{align}
\Gamma_1 = & \left(\Xi_1\mathcal{G}_{K,l}\zeta_{e,k}\right)^\top \mathcal{P}^{-1} \left(\Xi_1\mathcal{G}_{K,nl}(S(\xi_k)-S(\xi^*))\right) + \nonumber \\
& \left(\Xi_1\mathcal{G}_{K,nl}(S(\xi_k)-S(\xi^*))\right)^\top \mathcal{P}^{-1} \left(\Xi_1\mathcal{G}_{K,l}\zeta_{e,k}\right),
\end{align}
\begin{equation}
\Gamma_2 = \left(\Xi_1\mathcal{G}_{K,nl}(S(\xi_k)-S(\xi^*))\right)^\top \mathcal{P}^{-1} \left(\Xi_1\mathcal{G}_{K,nl}(S(\xi_k)-S(\xi^*))\right).
\end{equation}

In the context of Lemma~1, \( \Gamma_1 \) can be expressed as
\begin{equation}
\begin{aligned}
\Gamma_1 & \leq \tau \left(\Xi_1\mathcal{G}_{K,l}\zeta_{e,k}\right)^\top \mathcal{P}^{-1} \left(\Xi_1\mathcal{G}_{K,l}\zeta_{e,k}\right) + \\
& \varepsilon \tau^{-1}\left(\Xi_1\mathcal{G}_{K,nl}(S(\xi_k)-S(\xi^*))\right)^\top \left(\Xi_1\mathcal{G}_{K,nl}(S(\xi_k)-S(\xi^*))\right).
\end{aligned}
\end{equation}

Assuming \( S(\xi) \) is Lipschitz continuous with constant \( L \), the following inequality holds
\begin{equation}
\|S(\xi_k) - S(\xi^*)\| \leq L \|\xi_k - \xi^*\| = L \|\xi_{e,k}\|.
\end{equation}

By squaring both sides, we obtain
\begin{equation}
(S(\xi_k) - S(\xi^*))^\top (S(\xi_k) - S(\xi^*)) \leq L^2 (\xi_{e,k}^\top \xi_{e,k}).
\end{equation}

Given that \( \mathcal{Q}^\top \mathcal{Q} \) is positive definite, the following equivalent form can be used
\begin{equation}
\xi_{e,k}^\top \mathcal{Q}^\top \mathcal{Q} \xi_{e,k} \geq L^2 (\xi_{e,k}^\top \xi_{e,k}),
\end{equation}

and therefore
\begin{equation}
(S(\xi_k) - S(\xi^*))^\top (S(\xi_k) - S(\xi^*)) \leq \xi_{e,k}^\top \mathcal{Q}^\top \mathcal{Q} \xi_{e,k}.
\end{equation}

Given that
\begin{equation}
\begin{aligned}
& (S(\xi_k) - S(\xi^*))^T (\Xi_1\mathcal{G}_{K,nl})^T \Xi_1\mathcal{G}_{K,nl} (S(\xi_k) - S(\xi^*)) \leq \\
& \|\Xi_1\mathcal{G}_{K,nl}\|_2^2 \zeta_{e,k}^T \mathcal{Q}^T \mathcal{Q} \zeta_{e,k}.
\end{aligned}
\end{equation}

Since the goal is to minimize the norm of the nonlinear residuals bounded by \( \gamma \), it follows that
\begin{equation}\label{eq.proof_LMI_CE_2}
\|\Xi_1\mathcal{G}_{K,nl}\| \leq \gamma,
\end{equation}
one gets
\begin{equation}\label{eq.proof_LMI_CE_2_2}
\begin{aligned}
& (S(\xi_k)-S(\xi^*))^T (\Xi_1\mathcal{G}_{K,nl})^T \Xi_1\mathcal{G}_{K,nl} (S(\xi_k)-S(\xi^*)) \leq \\
& \zeta_{e,k}^T \gamma^2 \mathcal{Q}^T \mathcal{Q} \zeta_{e,k} \leq \zeta_{e,k}^T \mathcal{P}^{-1} \zeta_{e,k}.
\end{aligned}
\end{equation}

This further leads to the following constraint
\begin{equation}\label{eq.proof_LMI_CE_3}
\gamma^2 \mathcal{Q}^T \mathcal{Q} \preceq \mathcal{P}^{-1}.
\end{equation}

Then,
\begin{equation}
\begin{aligned}
\Gamma_1 & \leq \tau \left(\Xi_1\mathcal{G}_{K,l}\zeta_{e,k}\right)^\top \mathcal{P}^{-1} \left(\Xi_1\mathcal{G}_{K,l}\zeta_{e,k}\right) + \\
& \varepsilon \tau^{-1}\zeta_{e,k}^T \mathcal{P}^{-1} \zeta_{e,k}
\end{aligned}
\end{equation}
where \( \varepsilon \geq \lambda_{\max}(\mathcal{P}) \).

Moreover, considering that \( (\Xi_1\mathcal{G}_{K,nl}(S(\xi_k)-S(\xi^*)))^\top \mathcal{P}^{-1} (\Xi_1\mathcal{G}_{K,nl}(S(\xi_k)-S(\xi^*))) \leq \varepsilon (S(\xi_k)-S(\xi^*))^\top (\Xi_1\mathcal{G}_{K,nl})^\top \Xi_1\mathcal{G}_{K,nl} (S(\xi_k)-S(\xi^*)) \), and using the relations in \eqref{eq.proof_LMI_CE_2_2} and \eqref{eq.proof_LMI_CE_3}, an upper bound for \( \Gamma_2 \) can be derived as follows
\begin{equation}
\Gamma_2 \leq \varepsilon \zeta_{e,k}^\top \mathcal{P}^{-1} \zeta_{e,k}.
\end{equation}

Accordingly, inequality~\eqref{eq.proof_CE_5} can be simplified to
\begin{align}\label{eq.proof_CE_6}
& \zeta_{e,k}^\top \left[\lambda \mathcal{P}^{-1} - (1+\tau)(\Xi_1\mathcal{G}_{K,l})^\top \mathcal{P}^{-1} \Xi_1\mathcal{G}_{K,l} \right] \zeta_{e,k} \nonumber \\
& - \varepsilon (1+\tau^{-1}) \zeta_{e,k}^\top \mathcal{P}^{-1} \zeta_{e,k} \geq 0.
\end{align}

This condition is met if
\begin{equation}\label{eq.proof_CE_7}
(\lambda-\varepsilon (1+\tau^{-1})) \mathcal{P}^{-1} \succeq (1+\tau)(\Xi_1\mathcal{G}_{K,l})^\top \mathcal{P}^{-1} \Xi_1\mathcal{G}_{K,l}.
\end{equation}

By multiplying both sides by \( \mathcal{P} \) and applying the Schur complement, the following LMI is obtained
\begin{equation}\label{eq.LMI_CE}
\begin{bmatrix}
\lambda \mathcal{P} - \varepsilon(1+\tau^{-1}) \mathcal{P} & \sqrt{1+\tau}Y^\top \Xi_1^\top \\
(*) & \mathcal{P}
\end{bmatrix} \succeq 0.
\end{equation}

Here, \( \mathcal{Y} \triangleq \mathcal{G}_{K,l} \mathcal{P} \) is introduced to transform the bilinear matrix inequalities (BMIs) into LMIs. Additionally, by applying the Schur complement, the constraints \( \varepsilon \geq \lambda_{\max}(\mathcal{P}^{-1}) \), \eqref{eq.proof_LMI_CE_2}, and \eqref{eq.proof_LMI_CE_3} are reformulated as the LMIs \eqref{eq.DD_CE_conditions_4}, \eqref{eq.DD_CE_conditions_5}, and the first condition in \eqref{eq.DD_CE_conditions_7}, respectively. Moreover, constraints \eqref{eq.DD_CE_conditions_1} and \eqref{eq.DD_CE_conditions_2} directly follow from \eqref{eq.lemma_3}.

We now present the conditions required to ensure that the contractive ellipsoids are fully contained within the admissible set. Specifically, the ellipsoid \( \mathcal{E}(\mathcal{P}, 0) \) is contained within the polytope \( \mathcal{S} \) if and only if the following condition holds~\cite{boyd1994linear}

\begin{equation}\label{eq.set_containment}
\max \{\mathcal{F}_l\zeta_{e} \ | \ \zeta_e \in \mathcal{E}(\mathcal{P},0) \} \leq \bar{g}_l,
\end{equation}

which is equivalent to \eqref{eq.DD_CE_conditions_7}. This equivalence holds because the ellipsoidal set lies entirely within the polytope, implying that all points satisfying the ellipsoidal constraint also satisfy the polyhedral constraints
\begin{equation}\label{eq.set_containment_3}
\mathcal{F}_l\zeta_e \leq \bar{g}_l, \quad \mathrm{for} \,\,\, l=1,\ldots,q.
\end{equation}

Multiplying \eqref{eq.set_containment_3} by its transpose results in

\begin{equation}\label{eq.set_containment_4}
\mathcal{F}_l\zeta_e\zeta_e^T \mathcal{F}_l^T \leq \bar{g}_l^2.
\end{equation}

Additionally, by the definition of ellipsoidal sets, we have $\zeta_e\zeta_e^T \leq \mathcal{P}$. Therefore, inequality \eqref{eq.set_containment_4} becomes equivalent to

\begin{equation}\label{eq.set_containment_5_0}
\mathcal{F}_l\mathcal{P} \mathcal{F}_l^T \leq \bar{g}_l^2.
\end{equation}

Finally, applying the Schur complement to \eqref{eq.set_containment_5_0} leads to the second constraint in \eqref{eq.DD_CE_conditions_7}.

Additionally, to identify the largest ellipsoids, the standard method is to maximize their volume by maximizing the logarithm of the determinant of the shape matrix $\mathcal{P}$. This completes the proof.



\section{Proof of Theorem~\ref{thm:CHE}}\label{proof_thm_3}
Inspired by \cite{nguyen2023further}, we aim to show that if $\zeta_{e,k} \in \mathcal{S}_c$, then $\zeta_{e,k+1} \in \lambda_j' \mathcal{S}_c$. Since $\zeta_{e,k}$ belongs to the convex hull of ellipsoids, it can be expressed as
\begin{equation}
    \zeta_{e,k} = \sum_{j=1}^{n_e} \alpha_{j,k} \upsilon_{j,k},
\end{equation}
where $\upsilon_{j,k} \in \mathcal{E}(\mathcal{P}_j,0)$ and $\sum_{j=1}^{n_e} \alpha_{j,k} = 1$. 

To complete the proof, we need to show that if $\upsilon_{j,k} \in \mathcal{E}(\mathcal{P}_j,0)$, then $\upsilon_{j,k+1} \in \lambda\mathcal{S}_c$. To do so, by pre- and post-multiplying \eqref{eq.DD_CE_conditions_3_CHE} with
\begin{equation}\label{eq.temp_matrix_1_proof_1}
\begin{bmatrix}
I & 0 \\
0 & \mathcal{P}_j^{-1}
\end{bmatrix},
\end{equation}
we obtain
\begin{equation}\label{eq.openLoop_condition_1_changed}
\begin{bmatrix}
\mathcal{P}_i & \Xi_1\mathcal{G}_{K,l} \\
(*) & \lambda_j' \mathcal{P}_j^{-1}
\end{bmatrix} \succeq 0, \quad \forall i=1,\ldots,n_e.
\end{equation}

Multiplying \eqref{eq.openLoop_condition_1_changed} by $\alpha_{j,k}$ and summing over all $j$ results in
\begin{equation}\label{eq.openLoop_condition_1_changed_2}
\begin{bmatrix}
\sum\limits_{j=1}^{n_e}\alpha_{j,k}\mathcal{P}_j & \Xi_1\mathcal{G}_{K,l,j} \\
(*) & \lambda_j' \mathcal{P}_j^{-1}
\end{bmatrix} \succeq 0, \quad \forall j=1,\ldots,n_e.
\end{equation}

Using the Schur complement, equation \eqref{eq.openLoop_condition_1_changed_2} can be rewritten as
\begin{equation}\label{eq.Schur_proof_0}
(\Xi_1\mathcal{G}_{K,l,j})^\top \big(\sum\limits_{j=1}^{n_e}\alpha_{j,k}\mathcal{P}_j\big)^{-1}\Xi_1\mathcal{G}_{K,l,j} \leq \lambda_j' \mathcal{P}_j^{-1}.
\end{equation}
This implies, according to Corollary 1, that $\upsilon_j(t) \in \mathcal{E}(\mathcal{P}_j,0)$ leads to $\upsilon_j(t+1) \in \lambda_j'\mathcal{E}(\sum\limits_{j=1}^{n_e}\alpha_{j,k}\mathcal{P}_j,0)$.

Now, we show that $\lambda_j'\mathcal{E}(\sum\limits_{j=1}^{n_e}\alpha_{j,k}\mathcal{P}_j,0) \subseteq \lambda\mathcal{S}_c$. We proceed via contradiction. Suppose there exists a point $\zeta_{e,p} \in \lambda_j'\mathcal{E}(\sum\limits_{j=1}^{n_e}\alpha_{j,k}\mathcal{P}_j,0)$ that is not contained in the convex hull of ellipsoids. Without loss of generality, assume $\zeta_{e,p}$ lies on the boundary of $\lambda_j'\mathcal{E}(\sum\limits_{j=1}^{n_e}\alpha_{j,k}\mathcal{P}_j,0)$. Let $a_p \in \mathbb{R}^{n_e}$ define the supporting hyperplane at $\zeta_{e,p}$. Since both $\lambda_j'\mathcal{E}(\sum\limits_{j=1}^{n_e}\alpha_{j,k}\mathcal{P}_j,0)$ and $\lambda\mathcal{S}_c$ are symmetric about the origin, we obtain
\begin{equation}\label{eq.proof_supporting_hyperplane}
|a_p^T \zeta_e| < |a_p^T \zeta_{e,p}| = b_p^2, \quad \forall \zeta_e \in \lambda'\mathcal{S}_c.
\end{equation}
Thus, we have
\begin{equation}\label{eq.proof_supporting_hyperplane_2}
a_p^T \lambda'\mathcal{E}(\sum\limits_{j=1}^{n_e}\alpha_{j,k}\mathcal{P}_j,0) a_p = b_p^2.
\end{equation}

Since \eqref{eq.proof_supporting_hyperplane} must hold for all $\zeta_e \in \lambda'\mathcal{S}_c$, it follows that
\begin{equation}\label{eq.proof_supporting_hyperplane_4}
a_p^T \lambda_j' \mathcal{P}_j a_p < b_p^2.
\end{equation}
For any $\alpha_{j,k} \geq 0$ with $\sum\limits_{j=1}^{n_e} \alpha_{j,k} = 1$, we have
\begin{equation}\label{eq.proof_supporting_hyperplane_5}
a_p^T \lambda_j'\mathcal{E}(\sum\limits_{j=1}^{n_e}\alpha_{j,k}\mathcal{P}_j,0) a_p < b_p^2.
\end{equation}

This contradicts \eqref{eq.proof_supporting_hyperplane_2}. Hence, it is concluded that 
\[
\lambda_j'\mathcal{E}(\sum\limits_{j=1}^{n_e}\alpha_{j,t}\mathcal{P}_j,0) \subseteq \lambda\mathcal{S}_c.
\]
Using \eqref{eq.Schur_proof_0}, we get $\upsilon_{j,k+1} \in \lambda\mathcal{S}_c$, and thus $\zeta_{e,k+1} \in \lambda\mathcal{S}_c$.

The ellipsoid $\mathcal{E}(\mathcal{P}_i,0)$ is contained in the polytope $\mathcal{S}$ if and only if \cite{boyd1994linear}
\begin{equation}\label{eq.set_containment}
\max \{\mathcal{F}_l \zeta_e \mid \zeta_e \in \mathcal{E}(\mathcal{P}_i,0) \} \leq \bar{g}_l.
\end{equation}
Using Schur complement, this condition can be reformulated as
\begin{equation}\label{eq.set_containment_5}
\mathcal{F}_l \mathcal{P}_i \mathcal{F}_l^T \leq \bar{g}_l^2.
\end{equation}
Applying Schur complement again leads to constraint \eqref{eq.DD_CE_conditions_7_CHE}.

To maximize the convex hull, one typically maximizes the volume of ellipsoids. Alternatively, we can expand ellipsoidal shapes in specific reference directions, as discussed in \cite{hu2002analysis}. 

Let $d_i \in \mathbb{R}^{n_e}$ be the reference direction for $\mathcal{E}(\mathcal{P}_i,0)$. The optimization of $\mathcal{E}(\mathcal{P}_i,0)$ in this direction is equivalent to maximizing $\vartheta_i$ under
\begin{equation}
\vartheta_i^2 d_i^T \mathcal{P}_i^{-1} d_i \leq 1,
\end{equation}
which, via Schur complement, is reformulated as \eqref{eq.DD_CE_conditions_8_CHE}. This completes the proof.




\ifCLASSOPTIONcaptionsoff
  \newpage
\fi

\bibliographystyle{IEEEtran}
\bibliography{Refs}

\begin{thebibliography}{10}
\providecommand{\url}[1]{#1}
\csname url@samestyle\endcsname
\providecommand{\newblock}{\relax}
\providecommand{\bibinfo}[2]{#2}
\providecommand{\BIBentrySTDinterwordspacing}{\spaceskip=0pt\relax}
\providecommand{\BIBentryALTinterwordstretchfactor}{4}
\providecommand{\BIBentryALTinterwordspacing}{\spaceskip=\fontdimen2\font plus
\BIBentryALTinterwordstretchfactor\fontdimen3\font minus \fontdimen4\font\relax}
\providecommand{\BIBforeignlanguage}[2]{{%
\expandafter\ifx\csname l@#1\endcsname\relax
\typeout{** WARNING: IEEEtran.bst: No hyphenation pattern has been}%
\typeout{** loaded for the language `#1'. Using the pattern for}%
\typeout{** the default language instead.}%
\else
\language=\csname l@#1\endcsname
\fi
#2}}
\providecommand{\BIBdecl}{\relax}
\BIBdecl

\bibitem{gonzalez2015review}
D.~Gonz{\'a}lez, J.~P{\'e}rez, V.~Milan{\'e}s, and F.~Nashashibi, ``A review of motion planning techniques for automated vehicles,'' \emph{IEEE Transactions on intelligent transportation systems}, vol.~17, no.~4, pp. 1135--1145, 2015.

\bibitem{likhachev2003ara}
M.~Likhachev, G.~J. Gordon, and S.~Thrun, ``Ara*: Anytime a* with provable bounds on sub-optimality,'' \emph{Advances in neural information processing systems}, vol.~16, 2003.

\bibitem{lavalle2001randomized}
S.~M. LaValle and J.~J. Kuffner~Jr, ``Randomized kinodynamic planning,'' \emph{The international journal of robotics research}, vol.~20, no.~5, pp. 378--400, 2001.

\bibitem{mataric1992behavior}
M.~J. Mataric, ``Behavior-based control: Main properties and implications,'' in \emph{Proceedings, IEEE International Conference on Robotics and Automation, Workshop on Architectures for Intelligent Control Systems}.\hskip 1em plus 0.5em minus 0.4em\relax Citeseer, 1992, pp. 46--54.

\bibitem{sun2021motion}
H.~Sun, W.~Zhang, R.~Yu, and Y.~Zhang, ``Motion planning for mobile robots—focusing on deep reinforcement learning: A systematic review,'' \emph{IEEE Access}, vol.~9, pp. 69\,061--69\,081, 2021.

\bibitem{ge2002dynamic}
S.~S. Ge and Y.~J. Cui, ``Dynamic motion planning for mobile robots using potential field method,'' \emph{Autonomous robots}, vol.~13, pp. 207--222, 2002.

\bibitem{li2020hybrid}
H.~Li and P.~M. Wensing, ``Hybrid systems differential dynamic programming for whole-body motion planning of legged robots,'' \emph{IEEE Robotics and Automation Letters}, vol.~5, no.~4, pp. 5448--5455, 2020.

\bibitem{lavalle2006planning}
S.~M. LaValle, \emph{Planning algorithms}.\hskip 1em plus 0.5em minus 0.4em\relax Cambridge university press, 2006.

\bibitem{karaman2011sampling}
S.~Karaman and E.~Frazzoli, ``Sampling-based algorithms for optimal motion planning,'' \emph{The international journal of robotics research}, vol.~30, no.~7, pp. 846--894, 2011.

\bibitem{frazzoli2005maneuver}
E.~Frazzoli, M.~A. Dahleh, and E.~Feron, ``Maneuver-based motion planning for nonlinear systems with symmetries,'' \emph{IEEE transactions on robotics}, vol.~21, no.~6, pp. 1077--1091, 2005.

\bibitem{kavraki1996probabilistic}
L.~E. Kavraki, P.~Svestka, J.-C. Latombe, and M.~H. Overmars, ``Probabilistic roadmaps for path planning in high-dimensional configuration spaces,'' \emph{IEEE transactions on Robotics and Automation}, vol.~12, no.~4, pp. 566--580, 1996.

\bibitem{schouwenaars2001mixed}
T.~Schouwenaars, B.~De~Moor, E.~Feron, and J.~How, ``Mixed integer programming for multi-vehicle path planning,'' in \emph{2001 European control conference (ECC)}.\hskip 1em plus 0.5em minus 0.4em\relax IEEE, 2001, pp. 2603--2608.

\bibitem{weiss2017motion}
A.~Weiss, C.~Danielson, K.~Berntorp, I.~Kolmanovsky, and S.~Di~Cairano, ``Motion planning with invariant set trees,'' in \emph{2017 IEEE Conference on control technology and applications (CCTA)}.\hskip 1em plus 0.5em minus 0.4em\relax IEEE, 2017, pp. 1625--1630.

\bibitem{danielson2020robust}
C.~Danielson, K.~Berntorp, A.~Weiss, and S.~Di~Cairano, ``Robust motion planning for uncertain systems with disturbances using the invariant-set motion planner,'' \emph{IEEE Transactions on Automatic Control}, vol.~65, no.~10, pp. 4456--4463, 2020.

\bibitem{danielson2016path}
C.~Danielson, A.~Weiss, K.~Berntorp, and S.~Di~Cairano, ``Path planning using positive invariant sets,'' in \emph{2016 IEEE 55th Conference on Decision and Control (CDC)}.\hskip 1em plus 0.5em minus 0.4em\relax IEEE, 2016, pp. 5986--5991.

\bibitem{niknejad2024soda}
N.~Niknejad, R.~Esmzad, and H.~Modares, ``Soda-rrt: Safe optimal dynamics-aware motion planning,'' \emph{IFAC-PapersOnLine}, vol.~58, no.~28, pp. 863--868, 2024.

\bibitem{greiff2024invariant}
M.~Greiff, H.~Sinhmar, A.~Weiss, K.~Berntorp, and S.~Di~Cairano, ``Invariant set planning for quadrotors: Design, analysis, experiments,'' \emph{IEEE Transactions on Control Systems Technology}, 2024.

\bibitem{blanchini1999set}
F.~Blanchini, ``Set invariance in control,'' \emph{Automatica}, vol.~35, no.~11, pp. 1747--1767, 1999.

\bibitem{bisoffi2020data}
A.~Bisoffi, C.~De~Persis, and P.~Tesi, ``Data-based guarantees of set invariance properties,'' \emph{IFAC-PapersOnLine}, vol.~53, no.~2, pp. 3953--3958, 2020.

\bibitem{de2021low}
C.~De~Persis and P.~Tesi, ``Low-complexity learning of linear quadratic regulators from noisy data,'' \emph{Automatica}, vol. 128, p. 109548, 2021.

\bibitem{nguyen2023}
\BIBentryALTinterwordspacing
H.~N. Nguyen, ``{LMI conditions for robust invariance of the convex hull of ellipsoids with application to nonlinear state feedback control},'' Aug. 2023, working paper or preprint. [Online]. Available: \url{https://hal.science/hal-04177993}
\BIBentrySTDinterwordspacing

\bibitem{nguyen2024}
H.-N. Nguyen, ``Further results on the control law via the convex hull of ellipsoids,'' \emph{IEEE Transactions on Automatic Control}, vol.~69, no.~4, pp. 2753--2760, 2024.

\bibitem{hou2013model}
Z.-S. Hou and Z.~Wang, ``From model-based control to data-driven control: Survey, classification and perspective,'' \emph{Information Sciences}, vol. 235, pp. 3--35, 2013.

\bibitem{wang2016indirect}
Z.~Wang, R.~Lu, F.~Gao, and D.~Liu, ``An indirect data-driven method for trajectory tracking control of a class of nonlinear discrete-time systems,'' \emph{IEEE Transactions on Industrial Electronics}, vol.~64, no.~5, pp. 4121--4129, 2016.

\bibitem{bisoffi2022controller}
A.~Bisoffi, C.~De~Persis, and P.~Tesi, ``Controller design for robust invariance from noisy data,'' \emph{IEEE Transactions on Automatic Control}, vol.~68, no.~1, pp. 636--643, 2022.

\bibitem{modares2023data}
H.~Modares, ``Data-driven safe control of uncertain linear systems under aleatory uncertainty,'' \emph{IEEE Transactions on Automatic Control}, vol.~69, no.~1, pp. 551--558, 2023.

\bibitem{luppi2022data}
A.~Luppi, C.~De~Persis, and P.~Tesi, ``On data-driven stabilization of systems with nonlinearities satisfying quadratic constraints,'' \emph{Systems \& Control Letters}, vol. 163, p. 105206, 2022.

\bibitem{guo2023learning}
M.~Guo, C.~De~Persis, and P.~Tesi, ``Learning control of second-order systems via nonlinearity cancellation,'' in \emph{2023 62nd IEEE Conference on Decision and Control (CDC)}.\hskip 1em plus 0.5em minus 0.4em\relax IEEE, 2023, pp. 3055--3060.

\bibitem{de2023learning}
C.~De~Persis and P.~Tesi, ``Learning controllers for nonlinear systems from data,'' \emph{Annual Reviews in Control}, p. 100915, 2023.

\bibitem{guo2024data}
M.~Guo, C.~De~Persis, and P.~Tesi, ``Data-driven stabilization of nonlinear systems via taylor’s expansion,'' in \emph{Hybrid and Networked Dynamical Systems: Modeling, Analysis and Control}.\hskip 1em plus 0.5em minus 0.4em\relax Springer, 2024, pp. 273--299.

\bibitem{poursafar2010model}
N.~Poursafar, H.~Taghirad, and M.~Haeri, ``Model predictive control of non-linear discrete time systems: a linear matrix inequality approach,'' \emph{IET control theory \& applications}, vol.~4, no.~10, pp. 1922--1932, 2010.

\bibitem{modares2023}
H.~Modares, ``Data-driven safe control of uncertain linear systems under aleatory uncertainty,'' \emph{IEEE Transactions on Automatic Control}, pp. 1--8, 2023.

\bibitem{wang2019equivalence}
Z.~Wang, X.~Shen, and Y.~Zhu, ``On equivalence of major relaxation methods for minimum ellipsoid covering intersection of ellipsoids,'' \emph{Automatica}, vol. 103, pp. 337--345, 2019.

\bibitem{barber1996quickhull}
C.~B. Barber, D.~P. Dobkin, and H.~Huhdanpaa, ``The quickhull algorithm for convex hulls,'' \emph{ACM Transactions on Mathematical Software (TOMS)}, vol.~22, no.~4, pp. 469--483, 1996.

\bibitem{boyd1994linear}
S.~Boyd, L.~El~Ghaoui, E.~Feron, and V.~Balakrishnan, \emph{Linear matrix inequalities in system and control theory}.\hskip 1em plus 0.5em minus 0.4em\relax SIAM, 1994.

\bibitem{nguyen2023further}
H.-N. Nguyen, ``Further results on the control law via the convex hull of ellipsoids,'' \emph{IEEE Transactions on Automatic Control}, vol.~69, no.~4, pp. 2753--2760, 2023.

\bibitem{hu2002analysis}
T.~Hu, Z.~Lin, and B.~M. Chen, ``Analysis and design for discrete-time linear systems subject to actuator saturation,'' \emph{Systems \& control letters}, vol.~45, no.~2, pp. 97--112, 2002.

\end{thebibliography}

\vfill

\end{document}